\documentclass[mnsc]{informs3general}
\ifdefined\CHECKEQUATIONWIDTH\usepackage{eqndefns-left}\fi
\IfFileExists{tgtermes.sty}{\RequirePackage{tgtermes}}{}
\IfFileExists{newtxtext.sty}{\RequirePackage{newtxtext}}{}
\IfFileExists{newtxmath.sty}{\RequirePackage{newtxmath}}{}
\RequirePackage{bm}
\IfFileExists{endnotes.sty}{\RequirePackage{endnotes}}{}

\DoubleSpacedXI
\usepackage{natbib}
\usepackage{graphicx}
\usepackage{float}
\usepackage{tikz}
\usetikzlibrary{arrows.meta,positioning}
\bibpunct[, ]{(}{)}{,}{a}{}{,}
\def\bibfont{\small}

\EquationsNumberedThrough
\TheoremsNumberedThrough
\ECRepeatTheorems
\MANUSCRIPTNO{}

\usepackage{amsmath}
\usepackage{array}
\usepackage{enumitem}
\usepackage[unicode,colorlinks=true,citecolor=blue,linkcolor=black,urlcolor=blue,breaklinks=true]{hyperref}
\usepackage{booktabs}
\usepackage{threeparttable}

\DeclareMathOperator{\E}{E}
\DeclareMathOperator{\Var}{Var}

\DeclareMathOperator{\Prob}{P}
\DeclareMathOperator{\ATT}{ATT}

\newcommand{\eps}{\varepsilon}
\newcommand{\ind}{\mathbf 1}
\newcommand{\R}{\mathbb R}
\newcommand{\cP}{\mathcal P}

\renewcommand{\mathsf}[1]{\text{\sffamily #1}}

\ARTICLEAUTHORS{
\AUTHOR{Cheng Chao$^1$, Chenshan Hu$^2$, Yukai Huang$^3$}
\AFF{1. Department of Statistics and Data Science, Washington University in St. Louis}
\AFF{2. Leeds School of Business, University of Colorado Boulder}
\AFF{3. Sawyer Business School, Suffolk University}
} 

\usepackage{etoolbox}

\AtBeginEnvironment{table}{%
  \SingleSpacedXI
  \renewcommand{\arraystretch}{1}%
}

\AtBeginEnvironment{table*}{%
  \SingleSpacedXI
  \renewcommand{\arraystretch}{1}%
}

\makeatletter
\AtBeginDocument{%
  \pretocmd{\@makefntext}
    {\SingleSpacedXI\footnotesize}
    {}
    {\PackageWarning{spacingfix}{Footnote spacing patch failed}}%
}
\makeatother

\begin{document}

\RUNAUTHOR{}
\RUNTITLE{Separate versus Pooled Winsorization in DiD}
\TITLE{Separate versus Pooled Winsorization in Difference-in-Differences: Bias, Efficiency, and Finite-Sample Guarantees}

\ARTICLEAUTHORS{
\AUTHOR{Cheng Chao$^1$, Chenshan Hu$^2$, Yukai Huang$^3$}
\AFF{1. Department of Statistics and Data Science, Washington University in St. Louis}
\AFF{2. Leeds School of Business, University of Colorado Boulder}
\AFF{3. Sawyer Business School, Suffolk University}
} 


\ABSTRACT{%
\textit{\textbf{Problem Definition: }}
Extreme outcomes are common in operational data, and empirical
difference-in-differences (DiD) studies often winsorize them to limit
their influence. However, this practice affects not only the estimated
treatment effect but also its sampling uncertainty. We examine when
winsorization improves DiD estimation and whether cutoffs should be
computed within each group--period cell or from
the pooled sample.
\textit{\textbf{Methodology/Results: }}
We develop finite-sample and asymptotic analyses in a two-group,
two-period DiD setting. For continuous outcomes with finite variance, separate winsorization achieves a uniform sub-Gaussian deviation guarantee under a carefully calibrated winsorization fraction. However, neither raw DiD nor pooled winsorization can achieve the same guarantee. We also identify conditions under which separate winsorization is asymptotically unbiased, whereas pooled winsorization with a fixed fraction may remain biased. Under these conditions,
standard tail classifications alone do not establish an efficiency
ranking. Yet a two-sided decreasing-failure-rate (DFR) condition on
each cell's error distribution can guarantee strictly smaller asymptotic
variance for separate winsorization than for raw DiD at every fixed winsorization fraction. Numerical results agree with our finite-sample and asymptotic analyses. 
\textit{\textbf{Managerial Implications: }}
Our results provide practitioners with guidance on whether and how to winsorize outcomes in DiD studies. They clarify advantages and limitations of raw DiD, separate winsorization, and pooled winsorization. Moreover, the DFR condition introduced in our asymptotic efficiency analysis offers a simple diagnostic for when winsorization can be expected to achieve lower asymptotic variance while remaining unbiased for the treatment effect, giving theoretical grounding to a practice that has so far rested mainly on convention.
}


\maketitle

\section{Introduction}\label{sec:introduction}


Difference-in-differences (DiD) is one of the most widely used empirical methods for estimating the causal effects of policies, treatments, and interventions when randomized experiments are unavailable. As documented by \citet{currie2020technology}, nearly one-quarter of NBER working papers in applied microeconomics make references to DiD. In empirical operations management, DiD is viewed as one of five empirical tools for addressing endogeneity and self-selection \citep{ho2017causal}.
The applications of DiD span diverse operational settings \citep{lu2017mandatory,  cui2022gender, bekkerman2024impact, zheng2025video}.

Despite its widespread use, DiD estimation can be sensitive
to extreme outcomes in empirical data. This is a practical concern as skewed and heavy-tailed distributions are commonly encountered
in operational data, including service durations and product demand \citep{gans2003telephone, bimpikis2016inventory}. In its canonical form, DiD compares the change in the treated group's average outcome with the corresponding change in the control group. Because this comparison relies on sample means, extreme outcomes can generate substantial sampling variability, making the estimated treatment effect sensitive to a small number of observations. 
To limit their influence, empirical DiD studies often
apply \textit{winsorization} before estimation
\citep{ramani2019effects,kala2024managerial,lu2026lastmile}.
This procedure replaces values below the selected
lower quantile or above the selected upper quantile
with the corresponding cutoff value.

Although winsorization limits the influence of extreme observations, its consequences for the estimated treatment effect depend on how the cutoffs are constructed. A common practice is to winsorize the entire sample
using a single pair of cutoffs calculated from
observations pooled across groups and periods \citep{charoenwong2019regulatory,black2024sec,
kala2024managerial,bernstein2026ratings}. We call this approach \textit{pooled winsorization}. By contrast, a far less commonly used alternative calculates and applies cutoffs separately within each group--period cell. We call this approach \textit{separate winsorization}.  

The choice between these two winsorization methods is not merely a data-cleaning decision. Under pooled winsorization, common cutoffs can affect differently located group--period distributions unequally, whereas separate winsorization uses cell-specific cutoffs,
each calculated from a smaller sample. Both methods can introduce bias if the outcome distributions are heavily skewed, and pooled winsorization can distort mean differences even when each cell's distribution is symmetric. Moreover, although winsorization reduces the dispersion of the observed outcomes, estimating its cutoffs from empirical quantiles can introduce additional sampling uncertainty. Consequently, winsorization can potentially increase, rather than decrease, the variance of the resulting estimator.

To illustrate how winsorization affects the estimated
treatment effect, its standard error, and statistical
inference, we revisit the empirical analysis in \citet{charoenwong2019regulatory}. In 2012, a reform under the Dodd--Frank Act shifted
regulatory oversight of midsize investment-advisory firms
from the SEC to state regulators.
Using annual panel data, the authors compare changes
in complaints per 100 representatives from the
pre-reform period (2009--2011) to the post-reform
period (2012--2014) between firms that switched to
state oversight and firms that remained under SEC oversight.
Inspection of the released code shows that the authors use pooled winsorization.
Table~\ref{tab:motivating-winsorization-twot-order} reports the estimation results without additional
winsorization, with pooled winsorization, and with
separate winsorization\footnote{We adopt the specification in Column~(3) of Panel~A from Table~2 in \citet{charoenwong2019regulatory}. 
Their paper describes winsorization at 1\%, but the released code uses the pooled 1st and 98th percentiles. Following that implementation reproduces the reported coefficient of 0.028 after rounding. Our comparison uses the 1st and 99th percentiles for both winsorization rules.}.  
Both winsorization rules yield smaller estimates and
standard errors than the raw specification.
The raw estimate is 0.1370 ($p=0.1752$), compared with
0.0554 ($p=0.0301$) under pooled winsorization and
0.0471 ($p=0.0599$) under separate winsorization.
Among the three approaches, only the pooled estimate is statistically
significant at the $5\%$ level.
These differences illustrate sensitivity to the
winsorization rule but do not establish which procedure
more accurately estimates the original treatment effect.

\begin{table}[!t]
\footnotesize
\centering
\caption{DiD Estimates under Alternative Winsorization Rules}
\label{tab:motivating-winsorization-twot-order}

\begin{threeparttable}
\begin{tabular}{l*{3}{c}}
\hline\hline
& Raw & Pooled (1st/99th) & Separate (1st/99th) \\
\hline
DiD estimate
& 0.1370 & 0.0554 & 0.0471 \\
& (0.0996) & (0.0248) & (0.0245) \\
$p$-value
& 0.1752 & 0.0301 & 0.0599 \\
\hline
Observations
& 35,629 & 35,629 & 35,629 \\
\hline\hline
\end{tabular}

\begin{tablenotes}[flushleft]
\scriptsize
\item[] \textit{Notes:} This table reports DiD estimates for complaints
per 100 investment-adviser representatives under alternative
winsorization rules. Each specification includes
35,629 firm-year observations from 50 states.
Standard errors are clustered by state.
\end{tablenotes}
\end{threeparttable}
\end{table}

This sensitivity raises two practical questions for DiD estimation. When is winsorization preferable to using raw outcomes? When it is preferable, should the outcomes be winsorized using separate or pooled cutoffs before estimating DiD? We address these questions in the canonical two-group,
two-period DiD setting through two complementary  analyses in the finite-sample and asymptotic regimes.
In the finite-sample regime, we examine whether the raw,
pooled-winsorized, and separately winsorized DiD estimators
admit uniform sub-Gaussian deviation guarantees.
In the asymptotic regime, we analyze the asymptotic bias
of pooled and separate winsorization at a fixed
winsorization fraction.
Under conditions ensuring that separate winsorization
is asymptotically unbiased, we then compare its
asymptotic variance with that of raw DiD. Finally, we complement these theoretical analyses with Monte Carlo simulations that compare the finite-sample performance of the three estimators across a range of error distributions and winsorization fractions. 
The results show that separate winsorization can improve efficiency while remaining unbiased, whereas pooled winsorization may introduce substantial bias.

Our analyses yield three contributions: one
concerning finite-sample deviation guarantees and two
concerning asymptotic bias and efficiency.
First, we establish a uniform sub-Gaussian deviation guarantee for suitably calibrated separate winsorization that neither raw DiD nor pooled winsorization can provide over the same distribution class.
Second, we characterize the asymptotic bias of separate and pooled winsorization. We identify conditions under which separate winsorization is asymptotically unbiased, while pooled winsorization may remain biased because its common cutoffs can affect cells with different mean values unequally, generating biases that fail to cancel.
Third, given the fact that the pooled estimator is often biased, we continue to examine whether separate winsorization improves asymptotic efficiency relative to raw DiD. Under the conditions ensuring asymptotic unbiasedness, we derive the asymptotic variance of separate winsorization and compare it with that of raw DiD. We show that well-adopted tail classes, such as sub-Gaussian or heavy-tailed, are insufficient to establish an efficiency ranking between separate winsorization and raw DiD that holds across all winsorization fractions. We therefore introduce a classification based on a two-sided decreasing-failure-rate (DFR) condition that, to the best of our knowledge, is new to the study of winsorization efficiency. This classification yields a stronger, universal result that, under the conditions ensuring asymptotic unbiasedness, separate winsorization has strictly smaller asymptotic variance than raw DiD for every fixed winsorization fraction $\varepsilon\in(0,1/2)$.

The remainder of the paper is organized as follows. Section \ref{sec:literature} reviews the related literature and our contributions. Section~\ref{sec:setup} introduces the DiD setting and the estimators we study.
Section~\ref{sec:finite-sample} develops the finite-sample analysis of pooled and separate winsorization.
Section~\ref{sec:asymptotics} compares their asymptotic bias and examines the efficiency of separate winsorization relative to raw DiD.
Section~\ref{sec:numerics} presents our numerical analysis and Section~\ref{sec:discussion} concludes. 

\section{Literature Review}
\label{sec:literature}
\label{sec:related-literature}


Our work is related to two streams of literature.
The first concerns DiD methods and their applications
in operations management.
The second concerns winsorization,
including the statistical properties of winsorized
means and their implications for treatment effect
estimation.

\subsection{Difference-in-Differences in Operations Management}
\label{sec:literature-did}

First, our work contributes to the growing literature on DiD
methods and their applications in operations management.
DiD is widely used to estimate the causal effects of
policies and interventions from observational data
\citep{roth2023whats}.
In empirical operations management, it is an established
tool with applications spanning workforce management,
service delivery, and technology adoption
\citep{ho2017causal,mithas2022causality}.
For example, \citet{lu2017mandatory} examine how mandatory
overtime laws affect staffing decisions and care quality
in nursing homes.
\citet{cui2022gender} study how pandemic lockdowns affect
the gender gap in research productivity.
In digital commerce, researchers use DiD to examine
the effects of promised delivery speed on online purchases
\citep{cui2024sooner} and mobile payment adoption on
customers' credit-card spending \citep{xu2024mobile}.

Alongside these applications, methodological research
examines how research design and estimation choices
affect DiD estimates and inference.
This literature addresses serial correlation
\citep{bertrand2004much} and staggered adoption with
heterogeneous treatment effects
\citep{goodmanbacon2021difference,callaway2021difference}.
Addressing a challenge in operational settings,
\citet{tang2026debiased} develop a debiasing method for DiD
when operational decisions adapt to past outcomes.
Other recent studies examine the consequences of preprocessing
data before estimation.
\citet{ham2024benefits} characterize when matching on
pre-treatment outcomes reduces or increases bias.
\citet{ge2025bias} further analyze its effects on variance
and mean squared error, showing that reducing bias need
not improve overall estimation accuracy.

Our research contributes to this literature by identifying
conditions under which winsorization consistently estimates the original DiD effect and improves its statistical efficiency.
Extreme observations are common in operational data and can
generate substantial sampling variability because DiD relies
on sample means. We compare raw DiD with pooled and separate winsorization in the canonical two-group, two-period design, evaluating all three estimators against the original DiD effect. This complements research on DiD identification and inference
by examining how outcome distributions and preprocessing
choices affect estimation under the maintained identifying
assumptions.

\subsection{Winsorization}
\label{sec:literature-winsorization}

Our work also adds to the literature on the statistical properties of winsorization and its implications for treatment effect estimation. Extreme observations are a recurring concern in operational
data, including right-skewed service durations and
heavy-tailed product demand
\citep{gans2003telephone,bimpikis2016inventory}.
Winsorization is a common approach to limiting their
influence and has been used in DiD studies across
operations management, marketing, and economics
\citep{
lu2026lastmile,kannan2025digital,ramani2019effects,kala2024managerial}.
Its use in practice, however, does not by itself establish
whether it improves estimation accuracy.
Evaluating its consequences requires understanding both
the statistical properties of winsorized means and
whether winsorization preserves the original treatment
effect. Therefore, we review these two branches of the literature on winsorization. 

The first branch focuses on the statistical properties of winsorized means. Research on robust mean estimation studies how to estimate a population mean reliably when observations have heavy tails.
With appropriate tuning and over specified confidence ranges,
robust estimators can achieve sub-Gaussian error guarantees
under finite-variance assumptions
\citep{catoni2012challenging,devroye2016subgaussian}.
Winsorization provides one way to obtain these guarantees.
\citet{lugosi2021robust} develop a modified trimmed-mean
estimator whose univariate construction uses winsorization
with sample splitting.
\citet{kock2025winsorized} establish finite-sample bounds
for an empirical winsorized mean without sample splitting.
These results provide the one-sample foundations for
our finite-sample analysis.
A complementary line of research evaluates winsorized means
through their bias and estimation efficiency.
\citet{bickel1965some} studies the asymptotic efficiency
of robust location estimators, including winsorized means,
while \citet{rivest1994statistical} examines the statistical
properties of winsorized means under skewed distributions.
The influence-function framework provides a tool for
analyzing sensitivity and asymptotic variability
\citep{hampel1974influence}.
When cutoffs are estimated from empirical quantiles,
the variability of a winsorized mean reflects sampling
uncertainty in both the observations and the cutoffs.
Reducing the dispersion of the transformed observations
therefore does not by itself establish a reduction
in the variance of the estimator. Our research contributes to this literature by comparing separate and pooled winsorization in the finite-sample regime, asking whether either approach can yield a uniform sub-Gaussian guarantee. We show that suitably calibrated separate winsorization achieves such a guarantee, whereas neither raw DiD nor pooled winsorization can. To our knowledge, this is the first comparison of pooled and separate winsorization in the finite-sample regime.

The second branch focuses on winsorization in treatment effect estimation, especially the choice between pooled and separate winsorizations. To the best of our knowledge, the choice has received little theoretical attention, despite the widespread use of pooled winsorization in practice. The closest related work is \citet{wicker2026winsorizing}, who applies both winsorization procedures to two influential empirical studies in economics and shows that the resulting estimates of treatment effects in randomized experiments can differ substantially. 
Our research contributes to this literature by identifying conditions under which separate winsorization improves asymptotic efficiency relative to raw DiD.
We focus this efficiency comparison on separate winsorization and raw DiD because both are asymptotically unbiased under mild conditions, whereas pooled winsorization may remain biased. More importantly, we examine whether there is a practical sufficient condition that guides practitioners to use a distributional assumption directly to indicate the efficiency gain of separate winsorization over raw DiD estimation. We show that commonly adopted tail distributions, such as sub-Gaussian and heavy-tailed distributions, cannot guarantee uniform dominance, and we introduce a new class of two-sided DFR distributions that can achieve universal dominance.

\section{Difference-in-Differences and Winsorization}\label{sec:setup}

In this section, we introduce the difference-in-differences (DiD) framework and the three estimators considered in our analysis.  Section~\ref{sec:DiD_setting} presents the classical two-group, two-period setup and defines the raw DiD estimator. Section~\ref{sec:winsorization_application} defines the separate and pooled winsorization estimators.

\subsection{The DiD Model}
\label{sec:DiD_setting}

We consider the canonical two-group, two-period
(\(2\times2\)) DiD setup\footnote{Our analysis extends naturally to other more complex settings, such as those with multiple clusters or multiple periods.}. 
Let \(D_i\) indicate individual \(i\)'s group membership,
with \(D_i=1\) for the treated group and \(D_i=0\)
for the control group. There are two periods, a pre-treatment period \(t=0\) and a post-treatment period \(t=1\). The treated group receives treatment only in period \(1\), while the control group remains untreated in both periods. We observe each individual's outcome \(Y_{it}\) in both~periods.


Following the potential-outcomes framework
\citep{rubin1974estimating}, let \(Y_{it}(0,0)\) denote individual \(i\)'s potential outcome in period \(t\) if untreated in both periods, and let \(Y_{it}(0,1)\) denote its potential outcome in period \(t\) if untreated in period \(0\) and treated in period \(1\).
In each pair, the first entry is the treatment status in period \(0\) and the second is the treatment status in period \(1\), with \(1\) indicating treated and \(0\) untreated.
Under the standard no-interference condition, these
potential outcomes do not depend on other individuals'
treatment assignments.
To simplify notation, we write
\(Y_{it}(0)=Y_{it}(0,0)\) and \(Y_{it}(1)=Y_{it}(0,1)\).
We observe \(Y_{it}(1)\) for treated individuals and
\(Y_{it}(0)\) for control individuals, so
\(Y_{it}=D_iY_{it}(1)+(1-D_i)Y_{it}(0)\).
The treatment effect compares the two potential outcomes
for the same individual.
We focus on the average treatment effect on the treated (ATT)
in the post-treatment period,
\begin{equation}\label{eq:att}
\ATT=\E\{Y_{i1}(1)-Y_{i1}(0)\mid D_i=1\}.
\end{equation}



The difficulty in identifying \(\ATT\) is that
\(Y_{i1}(0)\), the outcome a treated individual would
have experienced without treatment, is unobserved.
DiD relies on assumptions that allow the mean of these
counterfactual outcomes to be inferred from the control
group's change in outcomes and the treated group's
pre-treatment outcomes.
The key assumption is parallel trends, which requires
the two groups' mean outcomes to have changed by the
same amount in the absence of treatment.

\begin{assumption}[Parallel trends]\label{ass:parallel}
The untreated potential outcomes satisfy
\[
 \E\{Y_{i1}(0)-Y_{i0}(0)\mid D_i=1\}
 =
 \E\{Y_{i1}(0)-Y_{i0}(0)\mid D_i=0\}.
\]
\end{assumption}

Another important assumption for identifying \(\ATT\)
is no anticipation, which requires treatment to have
no effect before its implementation.
Without this assumption, the treated group's change
in outcomes could reflect both the post-treatment
effect and anticipatory effects.

\begin{assumption}[No anticipation]\label{ass:no-anticipation}
For all individuals with \(D_i=1\), the potential
outcomes in period \(0\) satisfy
\[
Y_{i0}(1)=Y_{i0}(0).
\]
\end{assumption}
Let \(\mu_{gt}=\E(Y_{it}\mid D_i=g)\) for \(g,t\in\{0,1\}\).  Under Assumptions~\ref{ass:parallel} and ~\ref{ass:no-anticipation}, the missing treated-group counterfactual mean can be decomposed as
\[
 \begin{aligned}
 \E\{Y_{i1}(0)\mid D_i=1\}
 &={}
 \E\{Y_{i0}(0)\mid D_i=1\}
 +\E\{Y_{i1}(0)-Y_{i0}(0)\mid D_i=1\}\\
 &={}
 \mu_{10}
 +\underbrace{\E\{Y_{i1}(0)-Y_{i0}(0)\mid D_i=0\}}_{\text{parallel trends}}\\
 &={}
 \mu_{10}+(\mu_{01}-\mu_{00}).
 \end{aligned}
\]
Substituting this expression into \eqref{eq:att} gives
\begin{equation}\label{eq:DiD-pop}
 \ATT=(\mu_{11}-\mu_{10})-(\mu_{01}-\mu_{00}).
\end{equation}
Thus, DiD identifies \(\ATT\) by subtracting the control group's change in mean outcomes from the treated group's change. For brevity, we define $\Delta = (\mu_{11}-\mu_{10})-(\mu_{01}-\mu_{00})$, which is the ATT.

Let \(n_g=\sum_i\ind\{D_i=g\}\) and \(\bar Y_{gt}=n_g^{-1}\sum_{i:D_i=g}Y_{it}\) for \(g,t\in\{0,1\}\).  For each individual, define the observed change \(\Delta Y_i=Y_{i1}-Y_{i0}\) and its group-specific mean \(\overline{\Delta Y}_g=n_g^{-1}\sum_{i:D_i=g}\Delta Y_i\).
Replacing the four expectations in \eqref{eq:DiD-pop} with their corresponding sample means gives the raw DiD estimator
\begin{equation}\label{eq:raw-DiD}
 \widehat\Delta^{\mathrm{raw}}
 =\overline{\Delta Y}_{1}-\overline{\Delta Y}_{0}
 =\bar Y_{11}-\bar Y_{10}-\bar Y_{01}+\bar Y_{00}.
\end{equation}
The raw DiD estimator also admits an exact two-way fixed-effects (TWFE) representation.  Consider the unweighted OLS regression
\begin{equation}\label{eq:twfe}
 Y_{it}=\alpha_i+\lambda_t
 +\tau D_i\ind\{t=1\}+u_{it},
\end{equation}
where $\alpha_i$ denotes the individual fixed effect, $\lambda_t$ denotes the time fixed effect, and $\tau$ denotes the coefficient on the interaction between the treatment-group indicator and the post-treatment indicator. In this canonical \(2\times2\) DiD design, the OLS coefficient satisfies \(\widehat\tau=\widehat\Delta^{\mathrm{raw}}\).  For a detailed discussion of this equivalence and the associated estimation and inference procedures, see \citet{roth2023whats}.  


\subsection{Winsorization and Its Application in DiD}
\label{sec:winsorization_application}


In empirical DiD applications, skewed or heavy-tailed outcomes can make sample means sensitive to a small number of extreme observations.  Researchers often use winsorization to limit this influence \citep{kannan2025digital, lu2026lastmile}. This procedure replaces observations below the lower or above the upper sample quantiles with the corresponding cutoff values, leaving the remaining observations unchanged.  To formalize this transformation, define the winsorization map for fixed cutoffs \(a\le b\) as
\[
 \phi_{a,b}(y)=a\ind\{y<a\}+y\ind\{a\le y\le b\}+b\ind\{y>b\}.
\]
For a sample \(z_1,\ldots,z_m\), write \(z_{(1)}\le\cdots\le z_{(m)}\) for its order statistics.  Given \(0<\eps<1/2\), let
\[
 k_m=\lceil m\eps\rceil,
 \qquad \widehat a_m=z_{(k_m)},
 \qquad \widehat b_m=z_{(m+1-k_m)}.
\]
The empirical winsorized mean is defined by
\begin{equation}\label{eq:emp-win-mean}
 W_{m,\eps}(z_1,\ldots,z_m)
 =\frac{1}{m}\sum_{i=1}^{m}\phi_{\widehat a_m,\widehat b_m}(z_i).
\end{equation}
To describe winsorization within each group--period cell, we index observations by cell. For each group \(g\in\{0,1\}\), relabel the \(n_g\) individuals with \(D_i=g\) as \(i=1,\ldots,n_g\), and let \(Y_{gt,i}\) denote the observed outcome of the \(i\)th individual of group \(g\) in period \(t\in\{0,1\}\). Each individual thus contributes the outcome vector \(\boldsymbol Y_{g,i}=(Y_{g0,i},Y_{g1,i})^{\top}\), and the group--period cell \((g,t)\) consists of the \(n_g\) observations \(Y_{gt,1},\ldots,Y_{gt,n_g}\).

We consider two ways of applying winsorization in DiD: separate winsorization and pooled winsorization.  Separate winsorization estimates the lower and upper cutoffs separately within each group--period cell and applies the resulting transformation only to observations in that cell.  The corresponding DiD estimator is
\begin{equation}\label{eq:sep-fixed}
 \begin{split}
 \widehat\Delta_{\eps}^{\mathsf{sep}}
 ={}&W_{n_1,\eps}(Y_{11,1},\ldots,Y_{11,n_1})
 -W_{n_1,\eps}(Y_{10,1},\ldots,Y_{10,n_1})\\
 &-W_{n_0,\eps}(Y_{01,1},\ldots,Y_{01,n_0})
 +W_{n_0,\eps}(Y_{00,1},\ldots,Y_{00,n_0}).
 \end{split}
\end{equation}
Pooled winsorization instead stacks all \(N=2(n_0+n_1)\) observed outcomes and estimates a single pair of cutoffs from the combined sample.  If \(Y_{(1)}^{P}\le\cdots\le Y_{(N)}^{P}\) are the pooled order statistics and \(k_N=\lceil N\eps\rceil\), define the common cutoffs as $ \widehat a^{P}=Y_{(k_N)}^{P},
 \widehat b^{P}=Y_{(N+1-k_N)}^{P}.$ The pooled DiD estimator is then
\begin{equation}\label{eq:pool-fixed}
 \widehat\Delta_{\eps}^{\mathsf{pool}}
 =\frac{1}{n_1}\sum_{i=1}^{n_1}
 \{\phi_{\widehat a^{P},\widehat b^{P}}(Y_{11,i})
 -\phi_{\widehat a^{P},\widehat b^{P}}(Y_{10,i})\}
 -\frac{1}{n_0}\sum_{i=1}^{n_0}
 \{\phi_{\widehat a^{P},\widehat b^{P}}(Y_{01,i})
 -\phi_{\widehat a^{P},\widehat b^{P}}(Y_{00,i})\}.
\end{equation}
Equivalently, replacing \(Y_{it}\) in regression~\eqref{eq:twfe} with outcomes winsorized at the common pooled cutoffs yields estimator~\eqref{eq:pool-fixed}, whereas using outcomes winsorized at the cell-specific cutoffs yields estimator~\eqref{eq:sep-fixed}.  Thus the same regression can implement distinct statistical estimators depending on how the cutoffs are computed.

Winsorization is frequently used as a preprocessing step in empirical research, but it is often unclear whether the cutoffs should be estimated separately within each cell or from the pooled sample.  Next, we compare separate and pooled winsorization, along with the raw DiD estimator without winsorization, through finite-sample and asymptotic analyses.

\section{Finite-sample Theory for Winsorized DiD Estimators}\label{sec:finite-sample}
In this section, we compare the finite-sample performance of the raw, pooled
and separate DiD estimators. Section~\ref{sec:finite-sample-criterion} introduces a uniform sub-Gaussian deviation guarantee and shows that raw DiD cannot achieve it. Section~\ref{sec:finite-sample-winsorization} then examines whether pooled or separate winsorization can achieve this guarantee.

Following the notation defined in Section~\ref{sec:winsorization_application}, for each group
\(g\in\{0,1\}\) we observe the outcome vectors:
\[
 \boldsymbol Y_{g,i}=(Y_{g0,i},Y_{g1,i})^{\top},
 \qquad i=1,\ldots,n_g.
\]
For each fixed group--period cell \((g,t)\), we assume that the observations \(Y_{gt,1},\ldots,Y_{gt,n_g}\) are independent and identically
distributed. Throughout this section, we assume that the sample distribution within each group--period cell $(g,t)$ belongs to \(\cP(\sigma^2)\), the class of continuous distributions with variance at most \(\sigma^2\) defined as follows:
\begin{equation}\label{eq:P-sigma}
 \cP(\sigma^2)=\left\{P\;\middle|\;P\text{ is a continuous distribution of }Z\text{ on }\R\text{ and }\Var_P(Z)\le\sigma^2\right\}.
\end{equation}
\subsection{Finite-sample Analysis and Sub-Gaussian Estimator}
\label{sec:finite-sample-criterion}
We first examine the asymptotic behavior of raw DiD. Let \(\nu_g^2=\Var(Y_{g1,i}-Y_{g0,i})\) denote the variance of the within-individual change in group \(g\), assuming the outcome variables across group--period cells are independent, and set \(s_n^2=\nu_1^2/n_1+\nu_0^2/n_0\). The central limit theorem then gives: 
\begin{equation}\label{eq:raw-general-clt}
 \frac{\widehat\Delta^{\mathrm{raw}}-\Delta}{s_n}
 \ \Rightarrow\ N(0,1)
 \qquad\text{as }n_0,n_1\to\infty.
\end{equation}
Accordingly, for every fixed \(\delta\in(0,1)\), we have \(\Prob(|\widehat\Delta^{\mathrm{raw}}-\Delta|\le z_{1-\delta/2}s_n)\to1-\delta\).  Because \(\nu_g^2\le4\sigma^2\), it follows that \(s_n^2\le4\sigma^2(1/n_1+1/n_0)\le8\sigma^2/n_{\min}\), where \(n_{\min}=\min(n_0,n_1)\). 
Combining this inequality with \(z_{1-\delta/2}\le\sqrt{2\log(2/\delta)}\), we have the following inequality:
$$
\liminf_{n_0,n_1\to\infty}\Prob\left(|\widehat\Delta^{\mathrm{raw}}-\Delta|\le 4\sigma\sqrt{\frac{\log(2/\delta)}{n_{\min}}}\right)\ge1-\delta.
$$
This is an asymptotic probability bound that holds even for heavy-tailed outcome distributions. However, the same inequality no longer holds if we remove the condition that the sample size is growing sufficiently large.   Our goal is to obtain a finite-sample bound with the same radius scale \(\sqrt{\log(2/\delta)/n_{\min}}\). We formulate this objective in the probably approximately correct (PAC) framework \citep{valiant1984theory}.  In this framework, one seeks quantitative estimates of how the estimation error \(|\widehat\Delta-\Delta|\) scales with the confidence parameter \(\delta\) and sample size \(n_{\min}\). Adapting the definition of a sub-Gaussian sample-mean estimator from \citet{devroye2016subgaussian}, we define a sub-Gaussian DiD estimator as follows.
\begin{definition}[Finite-sample sub-Gaussian DiD estimator]
\label{def:finite-sample-sg}
Let \(\widehat\Delta_\delta\) be the DiD estimator. It is a finite-sample sub-Gaussian estimator if there exists a constant \(C>0\), such that for every sample size $n_{\min}$ and every \(\delta\), the following finite-sample bound holds under any cell distributions \(F_{gt}\in\cP(\sigma^2)\) :
\begin{equation}\label{eqn:sub-gaussian finite sample bound}
    \Prob\!\left(
|\widehat\Delta_\delta-\Delta|
>C\sigma\sqrt{\frac{\log(2/\delta)}{n_{\min}}}
\right)\le\delta.
\end{equation}

\end{definition}

 In particular, Definition~\ref{def:finite-sample-sg} asks whether an estimator can have a sub-Gaussian-order confidence radius even when the outcomes themselves are heavy-tailed. Next, Theorem \ref{thm:raw-impossibility} shows that the raw DiD estimator cannot meet this requirement.

\begin{theorem}\label{thm:raw-impossibility}
For every constant \(C>0\), there exist a balanced group size \(n_0=n_1=m\), a confidence level \(\delta\in(0,1/4)\), and  cell distributions \(F_{gt}\in\cP(\sigma^2)\), such that
$$\Prob\left(|\widehat\Delta^{\mathrm{raw}}-\Delta|>C\sigma\sqrt{\frac{\log(2/\delta)}{m}}\right)>\delta.$$
\end{theorem}

The proof is in Appendix~\ref{app:counterexamples}.  It is the DiD analogue of the empirical-mean lower bound in Proposition~6.2 of \citet{catoni2012challenging}, which is also used by \citet{devroye2016subgaussian} to explain why the empirical mean is not uniformly sub-Gaussian under only a variance bound.

\subsection{Finite-sample Analysis of Winsorization}
\label{sec:finite-sample-winsorization}
Given the limitation of the raw DiD estimator, a common remedy is to winsorize the outcome before computing the mean differences. For the winsorization estimator considered in this subsection, we allow \(\eps\) to be chosen as a function of the confidence level \(\delta\) and the sample size \(n_{\min}\).  We are interested in understanding whether pooled or separate winsorization can recover the desired finite-sample bound \eqref{eqn:sub-gaussian finite sample bound}  and, if so, how to choose the winsorization level.

\subsubsection{Pooled Winsorization}
We start with the pooled winsorization. The next theorem shows that pooled winsorization cannot recover the finite-sample bound \eqref{eqn:sub-gaussian finite sample bound} even under a balanced design with $n_0=n_1=m$.

\begin{theorem}[Impossibility result for pooled winsorization]
\label{thm:adaptive-pool-negative}
Under the balanced group design with  $n_0=n_1=m$, there exists $L>0$ such that the following statement holds. For every $\sigma^2>0$ and every deterministic rule
$
\eps:\mathbb{N}\times(0,\tfrac14)\to(0,\tfrac12),
$
for all $m\ge1$ and $\delta\in(0,\tfrac14)$, there exist cell 
distributions $F_{gt} \in \mathcal{P}(\sigma^2)$ 
such that
\begin{equation}
\label{eq:adaptive-pool-negative}
\Prob\!\left(
  \bigl|\hat\Delta^{\mathsf{pool}}-\Delta\bigr|
  >
  L\sigma\sqrt{\frac{\log(2/\delta)}{m}}
\right)
>
\delta.
\end{equation}
\end{theorem}

Theorem~\ref{thm:adaptive-pool-negative} establishes an impossibility result for pooled winsorization: regardless of how one designs the rule for selecting $\eps$, there always exists an instance on which the finite-sample bound \eqref{eqn:sub-gaussian finite sample bound} fails. Pooled winsorization therefore cannot yield a sub-Gaussian DiD estimator. To the best of our knowledge, this is the first impossibility result for pooled winsorization in the finite-sample analysis.

We briefly describe the proof idea for Theorem~\ref{thm:adaptive-pool-negative}. Within $\mathcal{P}(\sigma^2)$, one may choose two distributions whose means are sufficiently separated that the pooled order statistics segregate by group, so that the pooled cutoffs behave as group-specific extreme order statistics and winsorize the two groups asymmetrically. Any deterministic winsorization rule is then caught in an unavoidable dichotomy: weak winsorization fails to tame the heavy tail, while strong winsorization truncates a disproportionate share of the upper tail in one group, inducing substantial bias in its sample mean. Neither regime admits a sub-Gaussian bound. The full proof is presented in Appendix~\ref{app:pooled-adaptive}.

Moreover, Theorem~\ref{thm:adaptive-pool-negative} can be extended beyond the balanced group design with a straightforward adjustment, though we omit the details for brevity. In fact, such a restriction on the balanced group design is establishing a stronger impossibility result for pooled winsorization, because we could construct a more straightforward counterexample by considering an extremely imbalanced group design where the whole group with smaller size but higher average outcome is completely winsorized by the shared cutoff. 

\subsubsection{Separate Winsorization}

For separate winsorization, a carefully chosen value of \(\eps\) yields the following finite-sample guarantee in Theorem~\ref{thm:sep-winsor-mean}.
\begin{theorem}\label{thm:sep-winsor-mean}
Assume \(F_{gt}\in\mathcal{P}(\sigma^2)\) for
\(g,t\in\{0,1\}\). For any fixed \(c>1\), there exists
\(L>0\) such that, for all \(\delta\in(0,1)\), with
\(u_\delta=\log(24/\delta)\),
\(A_c=2c^2+\tfrac{8c}{3}\), and
\(\varepsilon=A_cu_\delta/n_{\min}\le1/4\),
\begin{equation}\label{eq:sep-winsor-mean-bound}
\Prob\!\left(
 |\widehat\Delta^{\mathsf{sep}}-\Delta|
 \le
 L\sigma\sqrt{\frac{\log(24/\delta)}{n_{\min}}}
\right)\ge1-\delta.
\end{equation}
\end{theorem}

Theorem~\ref{thm:sep-winsor-mean} is the natural counterpart to the
pooled-winsorization impossibility result in
Theorem~\ref{thm:adaptive-pool-negative}: whereas no deterministic rule for
the pooled level $\varepsilon$ can deliver a sub-Gaussian deviation of order
$\sigma\sqrt{\log(1/\delta)/n_{\min}}$ for $\hat\Delta^{\mathsf{pool}}$, there exists a deterministic rule for choosing $\varepsilon$ for $\hat\Delta^{\mathsf{sep}}$ that achieves such order. The proof is in Appendix~\ref{app:adaptive-separate}.
In Section~\ref{sec:numerics}, we present numerical evidence to further support Theorem~\ref{thm:sep-winsor-mean}.

The analysis in this section distinguishes separate from pooled
winsorization in the finite-sample regime using a PAC criterion.
We show that a suitably calibrated separate winsorization estimator
achieves a sub-Gaussian guarantee, whereas no deterministic choice
of the pooled winsorization fraction provides such a guarantee
uniformly. These results are particularly relevant when practitioners have limited
data and little knowledge of the underlying outcome distributions.
From a robust-statistical perspective, our analysis establishes
that separate winsorization attains the optimal-order confidence
radius over a broad distribution class.
When a larger sample size is available, a finer comparison becomes possible in the asymptotic regime.
In Section~\ref{sec:asymptotics}, we therefore examine the
asymptotic bias and efficiency of pooled winsorization, separate winsorization, and the raw DiD estimator.

\section{Asymptotic Analysis for Winsorized DiD Estimators}\label{sec:asymptotics}

In this section, we examine the asymptotic bias and
efficiency of winsorized DiD estimators. Section~\ref{sec:asymptotic-bias}  characterizes the bias
of separate and pooled winsorization. We identify conditions under which separate winsorization is asymptotically unbiased, whereas pooled winsorization may remain biased. Under these conditions, separate winsorization and raw DiD are both asymptotically unbiased for ATT. We therefore focus the efficiency comparison in
Section~\ref{sec:asymptotic-efficiency} on these two
estimators. Accounting for uncertainty from estimating the cutoffs,
we derive the asymptotic variance of separate winsorization
and identify a two-sided DFR condition under which it
is strictly smaller than that of raw DiD at every fixed
winsorization fraction.

To facilitate the asymptotic comparison, we consider the following
location-shift representation for the DiD outcomes:
\begin{equation}\label{eq:location-shift-asymptotic}
Y_{gt,i}=\mu_{gt}+e_{gt,i},
\qquad g,t\in\{0,1\},\quad i=1,\ldots,n_g.
\end{equation}
The errors \(e_{gt,i}\sim F_{gt}\) are mutually independent across individuals, groups, and periods, and each \(F_{gt}\in\cP(\sigma^2)\) has mean zero and a density \(f_{gt}\) that is strictly positive on its support.  Within each group--period cell, they are identically distributed, but their distributions may differ across groups and periods. We do not assume that the specific distributions \(F_{gt}\) are known to the practitioner in advance. 



\subsection{Bias of Separate and Pooled Winsorization}\label{sec:asymptotic-bias}

\subsubsection{A Unified Bias Decomposition}
Recall that under the identification assumptions stated in Section~\ref{sec:setup}, the DiD effect
$
 \Delta=(\mu_{11}-\mu_{10})-(\mu_{01}-\mu_{00})
$
equals \(\ATT\). In both separate and pooled winsorization, each outcome is replaced by \(\phi_{a,b}(\cdot)\) for some lower cutoff \(a\) and upper cutoff \(b\), where \(\phi_{a,b}\) is the winsorization map defined in Section~\ref{sec:winsorization_application}. The two methods differ only in how these cutoffs are chosen. This lets us describe the bias of both methods in one framework.

Formally, for \(j\in\{\mathsf{sep},\mathsf{pool}\}\) and a fixed winsorization fraction \(\eps\in(0,1/2)\), method \(j\) assigns to each cell \((g,t)\) a pair of population cutoffs
\[
 A_{gt}^j<B_{gt}^j,\qquad g,t\in\{0,1\},
\]
and winsorizes every observation in that cell at \([A_{gt}^j,B_{gt}^j]\). These cutoffs depend on \(\eps\) and on the cell distributions, and we suppress the dependence on \(\eps\) for brevity. Separate winsorization allows the cutoffs to vary across the four cells, so that \((A_{gt}^{\mathsf{sep}},B_{gt}^{\mathsf{sep}})\) is determined by cell \((g,t)\) alone. Pooled winsorization instead shares a common pair of cutoffs across all cells, with \(A_{gt}^{\mathsf{pool}}=A^{\mathsf{pool}}\) and \(B_{gt}^{\mathsf{pool}}=B^{\mathsf{pool}}\) for every \((g,t)\). The next two subsections specify these cutoffs for each method. Here we take them as given and derive a bias decomposition that holds for any choice of cell-specific cutoffs.

For fixed cutoffs \(a<b\), define the winsorized mean of the error term $e\equiv e_{gt,i}$ as
\begin{equation}\label{eq:beta-definition}
 \begin{split}
 \beta_{F_{gt}}(a,b)
 &=\E\{\phi_{a,b}(e)\}\\
 &=\underbrace{\E\{(a-e)\ind(e<a)\}}_{\text{increase from the lower tail}}
 -\underbrace{\E\{(e-b)\ind(e>b)\}}_{\text{decrease from the upper tail}}.
 \end{split}
\end{equation}
Because \(e\) has mean zero, \(\beta_{F_{gt}}(a,b)\) is the change in the cell mean caused by winsorizing the error at \([a,b]\). It vanishes exactly when the lower-tail increase balances the upper-tail decrease. Since \(\phi_{A,B}(\mu+e)=\mu+\phi_{A-\mu,B-\mu}(e)\), for any cutoffs \(A<B\), the location-shift representation \eqref{eq:location-shift-asymptotic} gives
\begin{equation}\label{eq:location-winsor-mean}
 \E\{\phi_{A,B}(Y_{gt,i})\}
 =\mu_{gt}+\beta_{F_{gt}}(A-\mu_{gt},B-\mu_{gt}).
\end{equation}

Applying \eqref{eq:location-winsor-mean} cell by cell with the cutoffs of method \(j\), we define the bias of method \(j\) in cell \((g,t)\) as
\[
 \beta_{gt}^j=\beta_{F_{gt}}(A_{gt}^j-\mu_{gt},B_{gt}^j-\mu_{gt}).
\]
Let \(\Delta_\eps^j\) denote the expected value of the winsorized estimators,
\[
 \begin{split}
 \Delta_\eps^j=
 \left[\E\{\phi_{A_{11}^j,B_{11}^j}(Y_{11,i})\}
       -\E\{\phi_{A_{10}^j,B_{10}^j}(Y_{10,i})\}\right]-\left[\E\{\phi_{A_{01}^j,B_{01}^j}(Y_{01,i})\}
       -\E\{\phi_{A_{00}^j,B_{00}^j}(Y_{00,i})\}\right],
 \end{split}
\]
which is the quantity that the winsorized estimator \(\widehat\Delta_\eps^j\) targets. We define the bias of method \(j\) as \(B_\eps^j=\Delta_\eps^j-\Delta\). Substituting \eqref{eq:location-winsor-mean} into the expression of $B_\eps^j$, 
we obtain
\begin{equation}\label{eq:unified-DiD-bias}
 B_\eps^j
 =\Delta_\eps^j-\Delta
 =(\beta_{11}^j-\beta_{10}^j)-(\beta_{01}^j-\beta_{00}^j).
\end{equation}

Thus, method \(j\) is unbiased if and only if \(\beta_{11}^j-\beta_{10}^j=\beta_{01}^j-\beta_{00}^j\). Since the two methods differ only in the cutoffs \((A_{gt}^j,B_{gt}^j)\), comparing their bias $B_\eps^j$ comes down to comparing the cell biases \(\beta_{gt}^j\). We next analyze the bias of separate and pooled winsorization in turn.

\begin{remark}\label{remark:fix cutoff for bias analysis}
Here, we analyze winsorization
bias using fixed population-quantile cutoffs $a$ and $b$.
In practice, these cutoffs are estimated by the empirical
quantiles $\hat{a}$ and $\hat{b}$, as described in
\eqref{eq:emp-win-mean}. Under the stated quantile regularity
conditions, the additional error from estimating the cutoffs
vanishes as the sample size increases and therefore does not
change the asymptotic bias. We use population cutoffs in that
section to simplify the illustration. However, this additional
error generally contributes to the asymptotic variance. Consequently, the variance analysis
in Section~\ref{sec:asymptotic-efficiency} explicitly accounts
for empirical-quantile estimation as a source of sampling
uncertainty.
\end{remark}

\subsubsection{Bias of Separate Winsorization}
Write \(\ell_{gt,\eps}=F_{gt}^{-1}(\eps)\) and
\(u_{gt,\eps}=F_{gt}^{-1}(1-\eps)\).  Since separate winsorization uses
each cell's own outcome quantiles, we have 
\[
 A_{gt}^{\mathsf{sep}}=\mu_{gt}+\ell_{gt,\eps},
 \qquad
 B_{gt}^{\mathsf{sep}}=\mu_{gt}+u_{gt,\eps}.
\]
Then, its cell bias reduces to
\begin{equation}\label{eq:sep-cell-bias}
 \beta_{gt}^{\mathsf{sep}}
 =\beta_{F_{gt}}(\ell_{gt,\eps},u_{gt,\eps}).
\end{equation}
Next, Proposition~\ref{prop:sep-unbiasedness} gives three
alternative conditions on these error distributions, any one of
which ensures that the separate winsorization estimator is
unbiased.
\begin{proposition}\label{prop:sep-unbiasedness}
If any one of the following conditions holds, then  \(B_\eps^{\mathsf{sep}}=0\).
\begin{enumerate}[label=\textup{(A\arabic*)},ref=A\arabic*,leftmargin=*,itemsep=2pt]
\item\label{cond:sep-symmetry}  \(F_{gt}\) is symmetric around zero for each \(g\in\{0,1\}\) and \(t\in\{0,1\}\).
\item\label{cond:sep-time-stable}  \(F_{g0}=F_{g1}\) for each \(g\in\{0,1\}\).
\item\label{cond:sep-group-stable} \(F_{0t}=F_{1t}\) for each \(t\in\{0,1\}\).
\end{enumerate}
\end{proposition}

The formal proof of Proposition~\ref{prop:sep-unbiasedness} can be found in Appendix~\ref{app:asymptotics}.

\subsubsection{Bias of Pooled Winsorization}
Separate winsorization adapts its cutoffs to the location of each
cell, while pooled winsorization instead applies one pair
of cutoffs to all four cells. 
Consequently, different cells
can be affected differently by the common cutoffs.

To define these cutoffs, recall that \(n_g\) denotes the number
of individuals in group \(g\). Set \(n=n_0+n_1\), and suppose
\(n_g/n\to p_g\in(0,1)\), where \(p_0+p_1=1\).
Each individual contributes two outcome observations, at
\(t=0,1\), so the pooled sample contains \(2n\) observations.
Cell \((g,t)\) contributes \(n_g\) observations and therefore
has limiting weight \(p_g/2\) in the pooled sample.
Define \(Y_P\) as a random outcome obtained by selecting
cell \((g,t)\) with probability \(p_g/2\) and then drawing
an outcome from that cell's distribution. Its distribution
function, \(F_P(y)=\Prob(Y_P\le y)\), is the limiting pooled
outcome distribution. Specifically, 
\begin{equation}\label{eq:pooled-population-cdf}
\begin{split}
 F_P(y)={}&
 \frac{p_0}{2}
 \{F_{00}(y-\mu_{00})+F_{01}(y-\mu_{01})\}+\frac{p_1}{2}
 \{F_{10}(y-\mu_{10})+F_{11}(y-\mu_{11})\}.
\end{split}
\end{equation}
Define the population pooled cutoffs as $a_P=F_P^{-1}(\eps)$ and $b_P=F_P^{-1}(1-\eps)$. 
Pooled winsorization uses
\(A_{gt}^{\mathsf{pool}}=a_P\) and
\(B_{gt}^{\mathsf{pool}}=b_P\) in every cell. 
We define the transformation bias of winsorization in cell \((g,t)\) as $ \beta_{gt}^{\mathsf{pool}}$:
\begin{equation}\label{eq:pool-target}
\begin{split}
 \beta_{gt}^{\mathsf{pool}}
 &=
 \E\{\phi_{a_P,b_P}(Y_{gt,i})\}-\mu_{gt}=
 \beta_{F_{gt}}(a_P-\mu_{gt},\,b_P-\mu_{gt}).
\end{split}
\end{equation}
Therefore, as derived in \eqref{eq:unified-DiD-bias}, 
pooled winsorization remains unbiased if and only if
\(\beta_{11}^{\mathsf{pool}}-\beta_{10}^{\mathsf{pool}}
=\beta_{01}^{\mathsf{pool}}-\beta_{00}^{\mathsf{pool}}\). In contrast with \eqref{eq:sep-cell-bias}, the cell location
\(\mu_{gt}\) in \eqref{eq:pool-target} does not disappear. Thus, the pooled winsorization can be biased even under the assumptions of Proposition \ref{prop:sep-unbiasedness}, in which separate winsorization is unbiased.
We illustrate a counterexample in the following Proposition~\ref{prop:pool-normal-bias}.
\begin{proposition}\label{prop:pool-normal-bias}
Suppose
\(e_{gt,i}\sim N(0,\sigma^2)\), with \(\sigma^2>0\),
\(\mu_{11}=d>0\), and
\(\mu_{10}=\mu_{01}=\mu_{00}=0\), so that $\Delta = d$. Then, for every fixed \(\eps\in(0,1/2)\), $ \Delta_\eps^{\mathsf{sep}}=d$ and $
 \Delta_\eps^{\mathsf{pool}}<d.$
\end{proposition}

The proof of this counterexample is presented in Appendix~\ref{app:asymptotics}. Notice that (\ref{cond:sep-symmetry})--(\ref{cond:sep-group-stable}) all hold under the normal error distribution instances in Proposition~\ref{prop:pool-normal-bias}. Thus, the pooled winsorization has a stronger condition required for unbiasedness compared to the separate winsorization. 

Figure~\ref{fig:separate-pooled-cutoffs} illustrates the example
with \(\sigma^2=1\), \(d=3\),
\(\eps=0.025\), and equal group sizes. In this example, the pooled
cutoffs lower the treated-post mean while raising the means of the three other cells. These changes do not cancel
when computing DiD, despite the common symmetric error
distribution.


\begin{figure}[H]
\centering
\includegraphics[width=0.97\textwidth]
{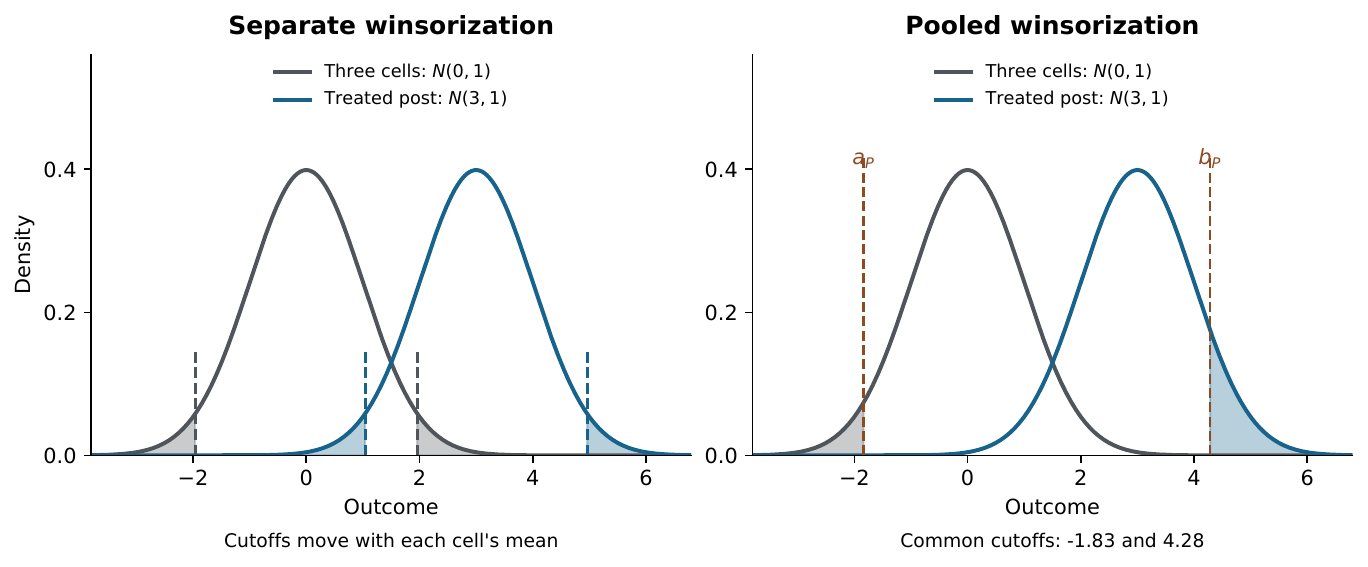}

\caption{DiD Bias: Separate versus Pooled Winsorization}
\label{fig:separate-pooled-cutoffs}

\vspace{0.6ex}

\begin{minipage}{0.97\textwidth}
\SingleSpacedXI
\scriptsize
\noindent
\textit{Notes:} This figure illustrates separate and pooled
winsorization at 2.5\% in each tail.
The control-pre, control-post, and treated-pre cells
all follow \(N(0,1)\) and are represented by the same
gray curve. The treated-post cell follows \(N(3,1)\)
and is shown in blue. Dashed lines mark the population cutoffs, and shaded
areas identify the affected tails. Separate cutoffs are symmetric about each cell mean and preserve all four cell means.
The common pooled cutoffs raise the means of the
control-pre, control-post, and treated-pre cells
and lower the treated-post mean, reducing the resulting
DiD below its original value of \(3\).
\par
\end{minipage}
\end{figure}

\subsection{Asymptotic Distribution and Efficiency Comparison}\label{sec:asymptotic-efficiency}
The bias analysis in Section~\ref{sec:asymptotic-bias} shows
that separate winsorization is asymptotically unbiased under
any one of conditions (\ref{cond:sep-symmetry}),
(\ref{cond:sep-time-stable}), or
(\ref{cond:sep-group-stable}) in Proposition \ref{prop:sep-unbiasedness}, whereas pooled winsorization
requires additional assumptions for unbiasedness. Motivated by this distinction, we next compare the asymptotic variance of the separate winsorization estimator
with that of the raw DiD estimator
$\hat{\Delta}^{\mathsf{raw}}$, which is unbiased for ATT. Specifically, we want to identify the underlying outcome distributions for which
separate winsorization has a smaller asymptotic variance
than raw~DiD.

We first derive the asymptotic variance of $\hat{\Delta}^{\mathsf{raw}}$. Recall that \(n=n_0+n_1\) with \(n_g/n\to p_g\in(0,1)\).  Write \(\sigma_{gt}^2=\Var(e_{gt,i})\). The central limit theorem gives:
\[
 \sqrt n(\widehat\Delta^{\mathrm{raw}}-\Delta)
 \ \Rightarrow\ N(0,V_{\mathrm{raw}}),
 \qquad
 V_{\mathrm{raw}}
 =\frac{\sigma_{11}^2+\sigma_{10}^2}{p_1}
 +\frac{\sigma_{01}^2+\sigma_{00}^2}{p_0}.
\]

Because the cutoffs are estimated from the same observations
used to calculate the winsorized mean, their sampling
uncertainty contributes to the estimator's asymptotic
variance. Suppress the subscript $\eps$ on the cutoffs and write $\beta_{gt}=\beta_{gt}^{\mathsf{sep}}$, as defined in \eqref{eq:sep-cell-bias}. Then, the influence function of the winsorized cell-specific mean can be written as
\begin{equation}\label{eq:cell-if-general}
\begin{split}
\psi_{gt,\eps}(e)
={}&\phi_{\ell_{gt},u_{gt}}(e)-\beta_{gt}\\
&+\eps\frac{\eps-\ind\{e\le\ell_{gt}\}}
                 {f_{gt}(\ell_{gt})}
+\eps\frac{1-\eps-\ind\{e\le u_{gt}\}}
                 {f_{gt}(u_{gt})}.
\end{split}
\end{equation}
Let
$v_{gt,\eps}=\Var\{\psi_{gt,\eps}(e_{gt,i})\}$, and we have the following theorem for the asymptotic variance of the separate winsorization estimator $\hat{\Delta}^{\mathsf{sep}}$.

\begin{theorem}\label{thm:sep-fixed-clt}
If any one of
(\ref{cond:sep-symmetry}),
(\ref{cond:sep-time-stable}), or
(\ref{cond:sep-group-stable}) holds, then
\[
\sqrt n
\bigl(\widehat\Delta_\eps^{\mathsf{sep}}-\Delta\bigr)
\ \Rightarrow\
N(0,V_{\mathsf{sep}}(\eps)),
\]
where
\begin{equation}\label{eq:vsep-general}
V_{\mathsf{sep}}(\eps)
=
\frac{v_{11,\eps}+v_{10,\eps}}{p_1}
+
\frac{v_{01,\eps}+v_{00,\eps}}{p_0}.
\end{equation}
\end{theorem}
The detailed proof of Theorem~\ref{thm:sep-fixed-clt} is
presented in Appendix~\ref{app:empirical-cutoff-clt}.
Building on this result, we compare
$V_{\mathsf{sep}}(\eps)$ and $V_{\mathsf{raw}}$ across
different classes of distributions. Our objective is
to identify an interpretable condition under which
separate winsorization improves asymptotic efficiency
relative to raw DiD, thereby informing the choice
between these estimators under different
data-generating processes.


To compare the efficiency of separate winsorization and
raw DiD, we compare their asymptotic variances
$V_{\mathsf{sep}}(\eps)$ and $V_{\mathsf{raw}}$
derived above. A common intuition is that winsorization
reduces the influence of extreme observations and therefore
may reduce the estimator's variance, particularly when
the underlying distribution has heavy tails. This raises
a natural question: can conventional classification
of tail distributions provide a sufficient
condition for an efficiency advantage?

In fact, conventional classifications on the tail distributions, such as sub-Gaussian or heavy tails, do not determine whether separate winsorization improves asymptotic
efficiency at an arbitrarily chosen winsorization fraction $\eps$. In Appendix~\ref{app:tail-class-proof}, we show that the tail class
alone does not determine the efficiency ranking. Heavy-tailed distributions do not guarantee that separate winsorization is more efficient than raw DiD at every winsorization fraction, while sub-Gaussian distributions do not guarantee that raw DiD is more efficient than separate winsorization at every fraction.
We next establish a more specific, two-sided DFR condition under
which separate winsorization has smaller asymptotic variance
than raw DiD for every fixed winsorization fraction.

\subsubsection{A Sufficient Condition for Efficiency Dominance of Winsorization}\label{sec:dfr-sufficient}

We introduce a class of distributions under
which separate winsorization is more efficient than raw DiD
for every fixed winsorization fraction. This class is defined
by a decreasing-failure-rate (DFR) condition, which is commonly used in operations management
\citep{lu2009incentives, mundt2022survival}.
Specifically, for a nonnegative random variable
\(X\) with density \(g(x)\) and positive survival function
\(S(x)=\Prob(X>x)\), define its failure rate as
\(\lambda_X(x)=g(x)/S(x)\). When \(X\) represents a duration,
$\lambda_X(x)$ describes the conditional chance per unit time
of the duration ending in the next small interval, given
that it has already exceeded \(x\). The DFR condition means
that this rate is nonincreasing.

DFR distributions are usually used to describe nonnegative
random variables, such as lifetimes or interarrival times.
However, in our location-shift model,
\(Y_{gt,i}=\mu_{gt}+e_{gt,i}\), the error
\(e_{gt,i}\) can take both positive and negative values.
We therefore consider a class of error distributions that
satisfy a DFR condition in both tails. In the following
definition, the two failure rates are measured outward
from the median. 

\begin{definition}\label{def:dfr-tail}
Let \(\mathcal P_{\mathrm{DFR}}(\sigma^2)\) consist of the
error distributions in \(\cP(\sigma^2)\) with a density
\(f\), a unique median \(m\), and \(F(y)\in(0,1)\) for every
\(y\in\R\), whose two outward failure rates
\begin{equation}\label{eq:global-asymmetric-dfr}
 \lambda_L(x)=\frac{f(m-x)}{F(m-x)},
 \qquad
 \lambda_R(x)=\frac{f(m+x)}{1-F(m+x)},
 \qquad x>0,
\end{equation}
are nonincreasing.
\end{definition}

The two-sided DFR condition encompasses a rich class of distributions. A familiar example is the Laplace distribution, which has been used to approximate empirical distributions of firm growth rates \citep{bottazzi2006explaining,arata2019firm} and wind-power forecast errors \citep{ji2018wind}. These applications illustrate the relevance of this class to business outcomes and operational uncertainty. A special case of Definition~\ref{def:dfr-tail} arises when
the error distribution is symmetric. Since the errors
have mean zero, symmetry implies \(m=0\). Both outward
failure rates then coincide with the failure rate of
\(|e_{gt,i}|\). Thus, in the symmetric case, the two-sided DFR
condition is equivalent to requiring that the nonnegative
random variable \(|e_{gt,i}|\) satisfy the usual DFR condition.

The following theorem shows that the two-sided
DFR condition guarantees smaller asymptotic variance for
separate winsorization than for raw DiD at every fixed
winsorization fraction $\eps\in(0,1/2)$.

\begin{theorem}\label{thm:global-dfr-dominance}
Suppose any one of (\ref{cond:sep-symmetry}), (\ref{cond:sep-time-stable}), or (\ref{cond:sep-group-stable}) holds.  If every cell distribution belongs to \(\mathcal P_{\mathrm{DFR}}(\sigma^2)\), then
\[
 V_{\mathsf{sep}}(\eps)<V_{\mathrm{raw}}
 \qquad\text{for every fixed }\eps\in(0,1/2).
\]
\end{theorem}

Theorem~\ref{thm:global-dfr-dominance} gives an interpretable
shape condition that guarantees strictly smaller asymptotic
variance of the separate winsorization for every fixed fraction. 
To our knowledge, the connection between two-sided DFR tails and the efficiency of separate winsorization is new. The detailed proof is
presented in Appendix~\ref{app:dfr-sufficient}.

%


\section{Numerical Analysis}\label{sec:numerics}
In this section, we use Monte Carlo simulations to compare
the raw estimator with the separate and pooled winsorization
estimators. Section~\ref{sec:simulation-design} describes
the simulation design and evaluation measures.
Section~\ref{sec:simulation-results} compares the estimators'
bias, root mean squared error (RMSE), and confidence interval coverage across error distributions and winsorization fractions.

\subsection{Simulation Design}\label{sec:simulation-design}

We consider a two-group, two-period DiD design. For group \(g\in\{0,1\}\), period \(t\in\{0,1\}\), and observation \(i\in\{1,\ldots,n_g\}\), outcomes are generated by $Y_{gt,i}=\mu_{gt}+e_{gt,i}$. We set $(\mu_{00},\mu_{01},\mu_{10},\mu_{11})=(0,0,0,1)$, so the true DiD effect is $\Delta=1$. The errors $e_{gt,i}$ are mutually independent, with mean 0 and unit variance. We consider a balanced design with \((n_0,n_1)=(1{,}000,1{,}000)\) and examine six error distributions. As summarized in Table~\ref{tab:simulation-distributions}, these distributions include sub-Gaussian, exponentially tailed, and heavy-tailed distributions, with and without the two-sided DFR property. 
We then compare the raw DiD estimator with the separate and
pooled winsorization estimators over the grid
\(\varepsilon\in
\{0.005,0.010,0.025, 0.0331,
0.040,0.050,0.075,0.100\}\).\footnote{The fraction \(\varepsilon_{\mathrm{FS}}\approx0.0331\)
is selected using the finite-sample rule for separate
winsorization in
Theorem~\ref{thm:sep-winsor-mean} by setting
\(c=1.1\), \(\delta=0.05\), and
\(n_{\min}=\min\{n_0,n_1\}=1{,}000\).}

\begin{table}[!t]
\footnotesize
\centering
\caption{Properties of the Distributions of the Error Term $e_{gt,i}$}
\label{tab:simulation-distributions}

\begin{threeparttable}
\setlength{\tabcolsep}{4pt}
\begin{tabular}{
    p{0.19\textwidth}
    p{0.43\textwidth}
    ccc}
\hline\hline
Distribution
& Definition
& Sub-Gaussian
& Heavy-tailed
& DFR \\
\hline
Normal
& \(N(0,1)\)
& Yes & No & No \\

Gaussian mixture
& \(\{0.95N(0,1)+0.05N(0,25)\}/\sqrt{2.2}\)
& Yes & No & No \\

Student-\(t_{2.5}\)
& \(t_{2.5}/\sqrt{5}\)
& No & Yes & No \\

Laplace
& \(\operatorname{Laplace}(0,1/\sqrt{2})\)
& No & No & Yes \\

Signed Lomax
& \(e_{gt,i}=SX\), where \(\Pr(S=1)=\Pr(S=-1)=1/2\) and
  \(\Pr(X>x)=(1+x)^{-3}\)
& No & Yes & Yes \\

Laplace mixture
& \(e_{gt,i}=Z/\sqrt{4.25}\), where
  \(Z\mid B\sim\operatorname{Laplace}(0,B)\) and
  \(\Pr(B=0.5)=\Pr(B=2)=1/2\)
& No & No & Yes \\
\hline\hline
\end{tabular}

\begin{tablenotes}[flushleft]
\scriptsize
\item[] \textit{Notes:} This table reports the construction and distributional properties of the six error distributions
used in the simulations. All distributions are standardized to have mean zero and variance one. DFR denotes the two-sided decreasing-failure-rate condition
in Definition~\ref{def:dfr-tail}.
\end{tablenotes}
\end{threeparttable}
\end{table}

Specifically, for each error distribution, we generate
\(R=1{,}000\) independent datasets.
Within each replication, we apply all three estimators
to the same dataset and evaluate both winsorization methods
at each fraction in the grid.
We report bias, root mean squared error (RMSE), and
empirical coverage of nominal 95\% confidence intervals.

For each estimator, we estimate its sampling variance as
\(\widehat V
=(\widehat v_{11}+\widehat v_{10})/n_1
+(\widehat v_{01}+\widehat v_{00})/n_0\).
For raw DiD, \(\widehat v_{gt}\) is the sample variance
of the outcomes in cell \((g,t)\).
For both winsorization methods, we construct naive intervals
using the within-cell sample variances of the winsorized
outcomes, treating the estimated cutoffs as fixed.
For separate winsorization, we also construct
influence-function-adjusted (IF-adjusted) intervals that
account for uncertainty from cutoff estimation.
In this case, \(\widehat v_{gt}\equiv\widehat v_{gt,\eps}\) is the sample variance
of the estimated influence-function values, where $v_{gt,\eps}$ is defined after \eqref{eq:cell-if-general}. 
To implement this adjustment, let \(\widehat Q_{gt}(p)\)
denote the empirical \(p\)-quantile of outcomes in
cell \((g,t)\).
We estimate the cutoff density \(f_{gt}\{F_{gt}^{-1}(p)\}\)
using
\(2h_p/\{\widehat Q_{gt}(p+h_p)-\widehat Q_{gt}(p-h_p)\}\),
where \(p\in\{\varepsilon,1-\varepsilon\}\) and
\(h_p=\min\{0.02,p/2,(1-p)/2\}\). 
We substitute the estimated densities and cutoffs into
the influence function in \eqref{eq:cell-if-general} for estimating $v_{gt}$. All confidence intervals take the form
\(\widehat\Delta\pm z_{0.975}\sqrt{\widehat V}\), where
\(z_{0.975}\) denotes the \(97.5\)th percentile of the standard normal distribution.

\subsection{Simulation Results}\label{sec:simulation-results}


We first compare the estimators at the theory-guided fraction
\(\varepsilon_{\mathrm{FS}}=0.0331\), examining bias, RMSE,
and confidence interval coverage. We then assess how these
results change across the full grid of winsorization fractions.

At \(\varepsilon_{\mathrm{FS}}\), separate winsorization
exhibits little bias and reduces RMSE relative to raw DiD
under five of the six error distributions (Table~\ref{tab:simulation-results}). Its maximum absolute bias is \(0.0022\), consistent with the symmetry condition for unbiasedness satisfied by all six distributions. Separate winsorization exhibits comparable or smaller RMSE than the raw DiD estimator. Compared to raw DiD, separate winsorization reduces RMSE by approximately
\(23\%\) under the Gaussian mixture, \(26\%\) under Student-\(t_{2.5}\), and \(27\%\) under Signed Lomax. In contrast, pooled winsorization has bias ranging
from \(-0.0643\) to \(-0.0589\) and a larger RMSE than both raw DiD and separate winsorization in all six scenarios. These findings favor separate over pooled winsorization in terms of both bias and RMSE.


We next investigate the 95\% confidence interval coverage. For separate winsorization, the IF-adjusted confidence intervals provide the coverage close to the nominal \(95\%\) level. At \(\varepsilon_{\mathrm{FS}}\), the naive confidence interval coverage ranges from \(86.5\%\) to \(93.2\%\), compared with \(94.6\%\) to \(96.3\%\) with IF-adjusted confidence intervals. The naive standard error treats the estimated winsorization cutoffs as fixed and therefore overlooks uncertainty from cutoff estimation. Table~\ref{tab:se-calibration-separate} shows the consequence. At the same fraction, average naive standard errors fall below the Monte Carlo standard deviations, producing confidence intervals that are too narrow. In contrast, IF-adjusted standard errors remain close to the Monte Carlo standard deviations, consistent with the improved coverage. For pooled winsorization, low coverage is accompanied by bias in the estimated treatment effect. Its naive coverage ranges from \(63.0\%\) to \(82.1\%\) at \(\varepsilon_{\mathrm{FS}}\). 

We further compare the performance of the three estimators by varying the winsorization fraction over the grid from 0.005 to 0.100 in Table~\ref{tab:simulation-sensitivity}.
The results show that separate winsorization has lower RMSE than raw DiD over all winsorization fractions for the three DFR distributions: Laplace, Signed Lomax, and the Laplace mixture. This pattern is consistent with the asymptotic efficiency result under two-sided DFR. 
Besides, separate winsorization reduces RMSE throughout the grid under Student-\(t_{2.5}\) distribution, but offers no substantial improvement under normal distribution. In addition, the small bias and near-nominal IF-adjusted coverage of separate winsorization persist across the winsorization fraction grid. Its average naive coverage across the six distributions falls from \(93.7\%\) at
\(\varepsilon=0.005\) to \(84.5\%\) at
\(\varepsilon=0.100\), whereas IF-adjusted coverage remains
close to \(95\%\). Accounting for cutoff estimation is
therefore important at larger fractions in
these simulations. Pooled winsorization becomes increasingly
biased as the fraction rises. Its average absolute bias increases from \(0.0101\) to \(0.1831\) and its confidence interval coverage falls from \(94.3\%\) to \(1.2\%\). Overall, separate winsorization maintains little bias and its IF-adjusted intervals achieve coverage close to 95\% throughout the grid, though its efficiency advantage over the raw DiD estimator depends on the error distribution and winsorization fraction.
\begin{table}[!t]
\footnotesize
\centering
\caption{Simulation Performance at
\(\varepsilon_{\mathrm{FS}}=0.0331\)}
\label{tab:simulation-results}

\begin{threeparttable}
\setlength{\tabcolsep}{4.9pt}
\begin{tabular}{l*{10}{r}}
\hline\hline
& \multicolumn{3}{c}{Raw}
& \multicolumn{4}{c}{Separate}
& \multicolumn{3}{c}{Pooled} \\
\cmidrule(lr){2-4}
\cmidrule(lr){5-8}
\cmidrule(lr){9-11}
Distribution
& Bias & RMSE & Cov.
& Bias & RMSE & \(\mathrm{Cov}_1\) & \(\mathrm{Cov}_2\)
& Bias & RMSE & Cov. \\
\hline
Normal
& \(-0.0016\) & \(0.0627\) & \(95.6\)
& \(-0.0015\) & \(0.0632\) & \(93.2\) & \(95.4\)
& \(-0.0607\) & \(0.0855\) & \(82.1\) \\

Gaussian mixture
& \(-0.0032\) & \(0.0637\) & \(94.9\)
& \(-0.0022\) & \(0.0490\) & \(91.6\) & \(96.0\)
& \(-0.0589\) & \(0.0748\) & \(74.8\) \\

Student-\(t_{2.5}\)
& \( 0.0010\) & \(0.0637\) & \(94.0\)
& \(-0.0009\) & \(0.0473\) & \(89.8\) & \(94.6\)
& \(-0.0612\) & \(0.0737\) & \(65.4\) \\

Laplace
& \( 0.0002\) & \(0.0621\) & \(96.2\)
& \( 0.0004\) & \(0.0615\) & \(92.6\) & \(96.3\)
& \(-0.0626\) & \(0.0835\) & \(79.4\) \\

Signed Lomax
& \(-0.0019\) & \(0.0667\) & \(93.8\)
& \( 0.0002\) & \(0.0490\) & \(86.5\) & \(95.9\)
& \(-0.0624\) & \(0.0742\) & \(63.0\) \\

Laplace mixture
& \(-0.0013\) & \(0.0631\) & \(95.2\)
& \(-0.0007\) & \(0.0612\) & \(88.5\) & \(95.6\)
& \(-0.0643\) & \(0.0818\) & \(75.7\) \\
\hline\hline
\end{tabular}

\begin{tablenotes}[flushleft]
\scriptsize
\item[] \textit{Notes:} This table reports the bias,
root mean squared error (RMSE), and confidence interval
coverage of the raw, separate, and pooled DiD estimators.
Results are based on 1,000 Monte Carlo replications with
\((n_0,n_1)=(1{,}000,1{,}000)\),
\(\varepsilon_{\mathrm{FS}}=0.0331\), and
the true effect \(\Delta=1\).
Coverage entries are percentages for nominal
\(95\%\) confidence intervals.
For raw DiD, Cov. reports coverage of the conventional
confidence interval.
For separate winsorization, \(\mathrm{Cov}_1\) and
\(\mathrm{Cov}_2\) report coverage of the naive and
influence-function-adjusted intervals, respectively.
For pooled winsorization, Cov. reports coverage of
the naive interval.
The naive intervals treat the estimated cutoffs as fixed,
whereas the influence-function-adjusted intervals
account for uncertainty from cutoff estimation.
\end{tablenotes}
\end{threeparttable}
\end{table}

\begin{table}[!t]
\footnotesize
\centering
\caption{Standard-Error Calibration for Separate Winsorization}
\label{tab:se-calibration-separate}

\begin{threeparttable}
\setlength{\tabcolsep}{3.3pt}
\begin{tabular}{l*{8}{r}}
\hline\hline
& \multicolumn{8}{c}{Winsorization fraction \(\varepsilon\)} \\
\cmidrule(lr){2-9}
Distribution
& 0.005 & 0.010 & 0.025 & 0.0331
& 0.040 & 0.050 & 0.075 & 0.100 \\
\hline
\multicolumn{9}{l}{\textit{Panel A: Naive standard error}} \\
Normal
& 1.002 & 0.991 & 0.959 & 0.941 & 0.929 & 0.910 & 0.857 & 0.809 \\
Gaussian mixture
& 0.909 & 0.825 & 0.900 & 0.900 & 0.891 & 0.873 & 0.831 & 0.793 \\
Student-\(t_{2.5}\)
& 0.909 & 0.889 & 0.842 & 0.823 & 0.810 & 0.797 & 0.758 & 0.725 \\
Laplace
& 0.999 & 0.975 & 0.919 & 0.893 & 0.875 & 0.851 & 0.794 & 0.741 \\
Signed Lomax
& 0.869 & 0.847 & 0.791 & 0.770 & 0.757 & 0.739 & 0.688 & 0.645 \\
Laplace mixture
& 0.969 & 0.940 & 0.861 & 0.817 & 0.798 & 0.769 & 0.697 & 0.635 \\
\hline
\multicolumn{9}{l}{\textit{Panel B: IF-adjusted standard error}} \\
Normal
& 1.011 & 1.012 & 1.009 & 1.010 & 1.011 & 1.008 & 1.004 & 1.004 \\
Gaussian mixture
& 1.013 & 1.011 & 1.085 & 1.057 & 1.034 & 1.005 & 1.000 & 1.004 \\
Student-\(t_{2.5}\)
& 1.014 & 1.013 & 1.007 & 1.010 & 1.007 & 0.997 & 1.000 & 1.005 \\
Laplace
& 1.027 & 1.026 & 1.023 & 1.026 & 1.025 & 1.020 & 1.025 & 1.024 \\
Signed Lomax
& 0.983 & 0.995 & 0.997 & 1.009 & 1.011 & 0.997 & 0.997 & 0.992 \\
Laplace mixture
& 1.015 & 1.021 & 1.025 & 1.023 & 1.031 & 1.028 & 1.039 & 1.028 \\
\hline\hline
\end{tabular}

\begin{tablenotes}[flushleft]
\scriptsize
\item[] \textit{Notes:} This table compares naive and
influence-function-adjusted (IF-adjusted) standard errors
for the separate winsorization estimator.
Results are based on 1,000 Monte Carlo replications.
Each entry is the average estimated standard error divided
by the Monte Carlo standard deviation of the DiD estimates.
A ratio of one indicates accurate standard-error calibration,
while a ratio below one indicates underestimated sampling
uncertainty. Naive standard errors treat the estimated
winsorization cutoffs as fixed, whereas IF-adjusted standard
errors account for uncertainty from cutoff estimation.
The theory-guided fraction is
\(\varepsilon_{\mathrm{FS}}=0.0331\).
\end{tablenotes}
\end{threeparttable}
\end{table}

\section{Discussion}\label{sec:discussion}
Extreme outcomes are common in operational data and can
generate substantial sampling variability in DiD estimation.
Researchers often use winsorization to limit their influence,
but its benefits depend on the underlying distributions and
how the cutoffs are constructed. In this paper, we examine
when winsorization is preferable to using raw outcomes and
whether separate or pooled winsorization should be used.
We address these questions in the canonical two-group,
two-period DiD setting through finite-sample and asymptotic
analyses.

Our finite-sample analysis establishes a uniform sub-Gaussian
deviation guarantee for separate winsorization over continuous
distributions with a common variance bound. The winsorization
fraction is calibrated to sample size and failure probability,
subject to the stated sample-size condition. Neither raw DiD
nor pooled winsorization with a deterministic fraction rule
provides the same uniform guarantee.

Our asymptotic analysis further shows that the construction
of cutoffs affects both bias and efficiency when the
winsorization fraction is fixed. We identify conditions under
which separate winsorization is asymptotically unbiased,
whereas pooled winsorization may remain biased. This
difference arises because common cutoffs can affect cells
with different means unequally, introducing changes that do
not cancel in the DiD estimate. We also examine whether
separate winsorization improves efficiency after accounting
for uncertainty from estimating the cutoffs. Under the
maintained sampling, regularity, and unbiasedness conditions,
two-sided decreasing failure rates ensure strictly smaller
asymptotic variance than raw DiD for every fixed winsorization
fraction $\varepsilon\in(0,1/2)$. Broad tail classifications
alone do not provide this guarantee across all fractions.

To summarize, we compare the finite-sample and asymptotic behavior between pooled and separate winsorization methods, suggesting that separate winsorization is the preferred approach. Moreover, this paper provides practical guidance on implementing winsorization in empirical DiD research.
First, when researchers have limited information about the outcome distributions, the finite-sample analysis in Theorem~\ref{thm:sep-winsor-mean} selects the winsorization fraction using only the sample size and desired confidence level. Under the theorem's conditions, this choice provides uniform finite-sample error control without requiring knowledge of the specific forms of the outcome distributions. Second, in the asymptotic regime, under our stated assumptions for unbiasedness, the two-sided DFR condition provides a sufficient condition under which separate winsorization is more efficient than the raw DiD estimator. Finally, based on the numerical results, confidence intervals for the separate winsorization estimator should use IF-adjusted standard errors, which account for the sampling uncertainty arising from the estimated winsorization cutoffs.




\bibliographystyle{informs2014}
\let\oldbibliography\thebibliography
\renewcommand{\thebibliography}[1]{%
 	\oldbibliography{#1}%
 	\baselineskip11pt 
 	\setlength{\itemsep}{1pt}
}
\typeout{} 
\bibliography{winsorization}
\newpage 

\begin{APPENDICES}
\renewcommand{\theHsection}{appendix.\Alph{section}}



\raggedbottom
\section{Additional Table}



\begin{table}[H]
\scriptsize
\centering
\caption{Simulation Performance across Winsorization Fractions}
\label{tab:simulation-sensitivity}
\begin{threeparttable}
\setlength{\tabcolsep}{3.6pt}
\renewcommand{\arraystretch}{0.90}
\begin{tabular}{lrrrrrrrrrr}
\hline\hline
& & \multicolumn{3}{c}{Absolute bias}
& \multicolumn{2}{c}{RMSE ratio}
& \multicolumn{4}{c}{Coverage (\%)} \\
\cmidrule(lr){3-5}
\cmidrule(lr){6-7}
\cmidrule(lr){8-11}
Distribution & \(\varepsilon\)
& Raw & Sep. & Pool
& Sep. & Pool
& Raw & \(\mathrm{Sep}_1\) & \(\mathrm{Sep}_2\) & Pool \\
\hline
Normal
& 0.0050 & 0.0016 & 0.0015 & 0.0092 & 0.998 & 1.005 & 95.6 & 95.8 & 95.9 & 95.0 \\
& 0.0100 & 0.0016 & 0.0015 & 0.0179 & 1.000 & 1.030 & 95.6 & 94.9 & 95.6 & 94.0 \\
& 0.0250 & 0.0016 & 0.0015 & 0.0452 & 1.006 & 1.210 & 95.6 & 93.6 & 95.6 & 87.8 \\
& 0.0331 & 0.0016 & 0.0015 & 0.0607 & 1.008 & 1.363 & 95.6 & 93.2 & 95.4 & 82.1 \\
& 0.0400 & 0.0016 & 0.0015 & 0.0735 & 1.009 & 1.508 & 95.6 & 92.7 & 95.6 & 74.8 \\
& 0.0500 & 0.0016 & 0.0015 & 0.0927 & 1.011 & 1.748 & 95.6 & 91.7 & 95.2 & 61.9 \\
& 0.0750 & 0.0016 & 0.0014 & 0.1413 & 1.022 & 2.425 & 95.6 & 90.5 & 94.9 & 28.0 \\
& 0.1000 & 0.0016 & 0.0016 & 0.1909 & 1.028 & 3.162 & 95.6 & 88.6 & 94.9 & 5.2 \\
\hline
Gaussian mixture
& 0.0050 & 0.0032 & 0.0031 & 0.0127 & 1.008 & 0.930 & 94.9 & 92.1 & 96.1 & 94.8 \\
& 0.0100 & 0.0032 & 0.0032 & 0.0221 & 1.008 & 0.893 & 94.9 & 89.6 & 95.9 & 93.3 \\
& 0.0250 & 0.0032 & 0.0026 & 0.0458 & 0.794 & 1.030 & 94.9 & 91.7 & 96.9 & 83.0 \\
& 0.0331 & 0.0032 & 0.0022 & 0.0589 & 0.769 & 1.174 & 94.9 & 91.6 & 96.0 & 74.8 \\
& 0.0400 & 0.0032 & 0.0023 & 0.0702 & 0.764 & 1.312 & 94.9 & 91.4 & 95.7 & 63.6 \\
& 0.0500 & 0.0032 & 0.0025 & 0.0874 & 0.759 & 1.539 & 94.9 & 90.4 & 94.9 & 44.3 \\
& 0.0750 & 0.0032 & 0.0026 & 0.1331 & 0.751 & 2.191 & 94.9 & 89.1 & 95.0 & 11.1 \\
& 0.1000 & 0.0032 & 0.0026 & 0.1813 & 0.742 & 2.914 & 94.9 & 88.1 & 94.5 & 0.9 \\
\hline
Student-\(t_{2.5}\)
& 0.0050 & 0.0010 & 0.0006 & 0.0103 & 0.847 & 0.790 & 94.0 & 92.3 & 95.1 & 93.4 \\
& 0.0100 & 0.0010 & 0.0011 & 0.0198 & 0.809 & 0.797 & 94.0 & 92.0 & 94.8 & 92.0 \\
& 0.0250 & 0.0010 & 0.0011 & 0.0469 & 0.762 & 0.995 & 94.0 & 89.9 & 94.3 & 78.0 \\
& 0.0331 & 0.0010 & 0.0009 & 0.0612 & 0.743 & 1.157 & 94.0 & 89.8 & 94.6 & 65.4 \\
& 0.0400 & 0.0010 & 0.0010 & 0.0726 & 0.734 & 1.301 & 94.0 & 89.2 & 94.1 & 53.7 \\
& 0.0500 & 0.0010 & 0.0007 & 0.0896 & 0.714 & 1.531 & 94.0 & 89.6 & 93.9 & 33.3 \\
& 0.0750 & 0.0010 & 0.0008 & 0.1314 & 0.681 & 2.139 & 94.0 & 86.6 & 94.4 & 4.2 \\
& 0.1000 & 0.0010 & 0.0010 & 0.1738 & 0.653 & 2.780 & 94.0 & 85.4 & 94.9 & 0.0 \\
\hline
Laplace
& 0.0050 & 0.0002 & 0.0002 & 0.0086 & 0.996 & 0.986 & 96.2 & 95.8 & 96.0 & 95.2 \\
& 0.0100 & 0.0002 & 0.0003 & 0.0181 & 0.997 & 1.003 & 96.2 & 95.1 & 96.2 & 94.2 \\
& 0.0250 & 0.0002 & 0.0005 & 0.0468 & 0.995 & 1.185 & 96.2 & 93.4 & 96.0 & 87.5 \\
& 0.0331 & 0.0002 & 0.0004 & 0.0626 & 0.991 & 1.345 & 96.2 & 92.6 & 96.3 & 79.4 \\
& 0.0400 & 0.0002 & 0.0003 & 0.0755 & 0.991 & 1.496 & 96.2 & 92.1 & 96.1 & 71.1 \\
& 0.0500 & 0.0002 & 0.0005 & 0.0946 & 0.983 & 1.742 & 96.2 & 91.0 & 95.4 & 56.2 \\
& 0.0750 & 0.0002 & 0.0006 & 0.1420 & 0.968 & 2.419 & 96.2 & 88.7 & 95.9 & 17.6 \\
& 0.1000 & 0.0002 & 0.0000 & 0.1896 & 0.954 & 3.141 & 96.2 & 85.5 & 95.3 & 0.8 \\
\hline
Signed Lomax
& 0.0050 & 0.0019 & 0.0008 & 0.0097 & 0.877 & 0.796 & 93.8 & 91.6 & 95.1 & 93.4 \\
& 0.0100 & 0.0019 & 0.0006 & 0.0191 & 0.828 & 0.779 & 93.8 & 90.4 & 95.8 & 91.9 \\
& 0.0250 & 0.0019 & 0.0001 & 0.0473 & 0.764 & 0.951 & 93.8 & 87.0 & 95.3 & 77.3 \\
& 0.0331 & 0.0019 & 0.0002 & 0.0624 & 0.734 & 1.112 & 93.8 & 86.5 & 95.9 & 63.0 \\
& 0.0400 & 0.0019 & 0.0004 & 0.0745 & 0.718 & 1.259 & 93.8 & 85.6 & 95.5 & 47.1 \\
& 0.0500 & 0.0019 & 0.0004 & 0.0921 & 0.692 & 1.487 & 93.8 & 84.5 & 95.4 & 26.0 \\
& 0.0750 & 0.0019 & 0.0006 & 0.1346 & 0.649 & 2.078 & 93.8 & 82.4 & 94.7 & 2.2 \\
& 0.1000 & 0.0019 & 0.0005 & 0.1754 & 0.614 & 2.669 & 93.8 & 79.8 & 94.5 & 0.0 \\
\hline
Laplace mixture
& 0.0050 & 0.0013 & 0.0014 & 0.0103 & 0.994 & 0.969 & 95.2 & 94.3 & 95.4 & 94.0 \\
& 0.0100 & 0.0013 & 0.0014 & 0.0198 & 0.986 & 0.973 & 95.2 & 92.9 & 95.7 & 93.0 \\
& 0.0250 & 0.0013 & 0.0010 & 0.0484 & 0.974 & 1.137 & 95.2 & 89.9 & 95.9 & 84.9 \\
& 0.0331 & 0.0013 & 0.0007 & 0.0643 & 0.971 & 1.297 & 95.2 & 88.5 & 95.6 & 75.7 \\
& 0.0400 & 0.0013 & 0.0006 & 0.0774 & 0.959 & 1.449 & 95.2 & 86.8 & 95.7 & 65.0 \\
& 0.0500 & 0.0013 & 0.0001 & 0.0965 & 0.939 & 1.696 & 95.2 & 85.8 & 95.7 & 45.6 \\
& 0.0750 & 0.0013 & 0.0000 & 0.1431 & 0.897 & 2.359 & 95.2 & 83.4 & 95.6 & 5.8 \\
& 0.1000 & 0.0013 & 0.0003 & 0.1876 & 0.857 & 3.031 & 95.2 & 79.8 & 95.5 & 0.0 \\
\hline\hline
\end{tabular}

\begin{tablenotes}[flushleft]
\scriptsize
\item[] \textit{Notes:} This table reports absolute bias, RMSE ratios,
and confidence interval coverage across winsorization
fractions. Results are based on 1{,}000 Monte Carlo replications of
the balanced design \((n_0,n_1)=(1{,}000,1{,}000)\), with \(\Delta=1\).
\(\mathrm{Sep}_1\) and \(\mathrm{Sep}_2\) denote naive and
IF-adjusted coverage for separate winsorization, respectively. The RMSE ratio is the RMSE of the separate or pooled winsorization estimator divided by that of the raw estimator; values below one favor the corresponding winsorization estimator. Pool denotes pooled winsorization with the naive confidence interval.
Raw absolute bias and coverage do not depend on \(\varepsilon\) and
are repeated for comparison. The theory-guided fraction is
\(\varepsilon_{\mathrm{FS}}=0.0331\).
\end{tablenotes}
\end{threeparttable}
\end{table}

\section{Proofs for Section~\ref{sec:finite-sample}}
\subsection{Proof of Theorem~\ref{thm:raw-impossibility}}
\label{app:counterexamples}

\proof{Proof.}
Fix \(C>0\) and \(\sigma^2>0\). Choose
\(\delta\in(0,1/4)\) sufficiently small that
\[
 \frac1{2\sqrt{\delta}}>C\sqrt{\log(2/\delta)},
\]
and fix any integer \(m\ge1\). Set
\[
 q=\frac{\delta}{2m},\qquad
 \eta=\min\left\{\frac{\sigma}{100},
                  \frac{\sigma}{100\sqrt{m\delta}}\right\},
 \qquad B=\sqrt{\frac{\sigma^2-\eta^2/3}{q}}.
\]
Assume that the outcomes \(Y_{gt,i}\) in all four
group--period cells follow the same distribution \(P\),
which is a mixture of three uniform distributions on
three distinct intervals. Specifically, an observation
from \(P\) follows
\(\operatorname{Unif}[B-\eta,B+\eta]\) with probability
\(q/2\),
\(\operatorname{Unif}[-B-\eta,-B+\eta]\) with probability
\(q/2\), and
\(\operatorname{Unif}[-\eta,\eta]\) with probability
\(1-q\).
This distribution is continuous, with mean zero and variance
\(qB^2+\eta^2/3=\sigma^2\), so \(P\in\cP(\sigma^2)\).
Under this distributional specification, we have
\(\Delta=0\).

Now draw \(m\) observations from each cell, with all
\(4m\) observations mutually independent, and consider the event
$$A = \left\{\text{Exactly one observation lies in either } \operatorname{Unif}[B-\eta,B+\eta]\text{ or }\operatorname{Unif}[-B-\eta,-B+\eta]\text{ regime} \right\}$$
The probability of \(A\) satisfies
\[
 \Prob(A)=4mq(1-q)^{4m-1}
 \ge2\delta\{1-(4m-1)q\}
 >2\delta(1-2\delta)>\delta.
\]
On \(A\), consider the sample mean estimator
\(\hat Y_{gt}\) for each group--period pair
\(g,t\in\{0,1\}\).
Three of the four sample means are calculated entirely
from observations drawn from
\(\operatorname{Unif}[-\eta,\eta]\), while the remaining
sample mean contains one extreme observation drawn from
either \(\operatorname{Unif}[B-\eta,B+\eta]\) or
\(\operatorname{Unif}[-B-\eta,-B+\eta]\).
The center of this extreme component contributes
\(B/m\) in absolute magnitude to the deviation from the
true DiD effect. Accounting for the uniform perturbations
in all four sample means, we therefore have, on \(A\),
\[
 |\widehat\Delta^{\mathrm{raw}}-\Delta|
 \ge\frac Bm-4\eta.
\]
By our choice of parameters,
\[
 \frac Bm\ge\frac{99}{100}\sigma\sqrt{\frac{2}{m\delta}},
 \qquad 4\eta\le\frac{\sigma}{25\sqrt{m\delta}}.
\]
Therefore, on \(A\),
\[
 |\widehat\Delta^{\mathrm{raw}}-\Delta|
 \ge\frac Bm-4\eta
 >\frac{\sigma}{2\sqrt{m\delta}}
 >C\sigma\sqrt{\frac{\log(2/\delta)}{m}}.
\]
Since \(\Prob(A)>\delta\), the theorem follows.
\Halmos\endproof

\subsection{Proof of Theorem~\ref{thm:adaptive-pool-negative}}
\label{app:pooled-adaptive}

\proof{Proof.}
Fix \(\sigma^2>0\), a deterministic rule \(\varepsilon\),
and arbitrary \(m\ge1\) and \(\delta\in(0,1/4)\).
Write \(\varepsilon=\varepsilon(m,\delta)\).
The pooled lower cutoff has rank
\(k=\lceil4m\varepsilon\rceil\in\{1,\ldots,2m\}\),
and the upper cutoff has rank \(4m+1-k\).
Set
\[
 q=\frac{\delta}{m},\qquad
 \eta=\min\left\{\frac{\sigma}{100},
                  \frac{\sigma}{100\sqrt{2m\delta}}\right\},
 \qquad
 B=\sqrt{\frac{\sigma^2-\eta^2/3}{q(1-q)}},
 \qquad D=B+4\eta.
\]
Let the outcomes in cells \((1,1)\) and \((0,0)\)
follow the same mixture distribution
\[
 P_+=(1-q)\operatorname{Unif}[-\eta,\eta]
       +q\operatorname{Unif}[-B-\eta,-B+\eta],
\]
and let the outcomes in cells \((1,0)\) and \((0,1)\)
follow \(P_-=\operatorname{Unif}[-D-\eta,-D+\eta]\).
Both distributions are continuous, with means
\(\mu_+=-qB\) and \(\mu_-=-D\), and variances
\(\sigma^2\) and \(\eta^2/3\), respectively.
Thus all four cell distributions belong to
\(\cP(\sigma^2)\), and the population DiD effect is
\(\Delta=2(\mu_+-\mu_-)=2(D-qB)\).
Draw \(m\) observations from each cell, with all
\(4m\) observations mutually independent.

The entire support of \(P_-\) lies below that of \(P_+\).
Consequently, for every \(k\le2m\), the pooled lower
cutoff lies among the observations from \(P_-\), while
the upper cutoff lies among the observations from \(P_+\).
Lower winsorization therefore leaves the observations
from \(P_+\) unchanged, and upper winsorization can
only decrease them. Moreover, every winsorized observation
from \(P_-\) remains in \([-D-\eta,-D+\eta]\).

Consider the event
\[
A=\left\{
\text{Exactly one observation across cells }(1,1)\text{ and }(0,0)
\text{ lies in the }\operatorname{Unif}[-B-\eta,-B+\eta]
\text{ regime}
\right\}.
\]
By Bernoulli's inequality,
\[
 \Prob(A)=2mq(1-q)^{2m-1}
 \ge2\delta\{1-(2m-1)q\}
 >2\delta(1-2\delta)>\delta.
\]

Let \(\hat Y_{gt}=m^{-1}\sum_{i=1}^m Y_{gt,i}\) and
\(\hat Y_{gt}^{\mathsf{pool}}
=m^{-1}\sum_{i=1}^m\phi_{\hat a^P,\hat b^P}(Y_{gt,i})\)
denote the sample means before and after pooled winsorization,
respectively. Write \(\mu_{gt}=\E(Y_{gt,i})\) for the original
population mean in cell \((g,t)\). Under our construction,
\(\mu_{11}=\mu_{00}=-qB\) and
\(\mu_{10}=\mu_{01}=-D\), so
\(\Delta=\mu_{11}-\mu_{10}-\mu_{01}+\mu_{00}=2D-2qB\).

On \(A\), exactly one observation across cells \((1,1)\)
and \((0,0)\) lies in the
\(\operatorname{Unif}[-B-\eta,-B+\eta]\) regime,
while all remaining observations in these two cells lie
in \([-\eta,\eta]\). Thus,
\(\hat Y_{11}+\hat Y_{00}\le2\eta-B/m\).
Because pooled winsorization can only decrease the observations
in these two cells,
\[
 \hat Y_{11}^{\mathsf{pool}}+\hat Y_{00}^{\mathsf{pool}}
 \le\hat Y_{11}+\hat Y_{00}
 \le2\eta-\frac Bm.
\]
Moreover, every winsorized observation in cells \((1,0)\)
and \((0,1)\) remains in \([-D-\eta,-D+\eta]\), giving
\(\hat Y_{10}^{\mathsf{pool}}\ge-D-\eta\) and
\(\hat Y_{01}^{\mathsf{pool}}\ge-D-\eta\).
Consequently, on \(A\),
\[
 \begin{split}
 \widehat\Delta^{\mathsf{pool}}-\Delta
 &=\hat Y_{11}^{\mathsf{pool}}
   -\hat Y_{10}^{\mathsf{pool}}
   -\hat Y_{01}^{\mathsf{pool}}
   +\hat Y_{00}^{\mathsf{pool}}-(2D-2qB)\\
 &\le4\eta-B(1/m-2q)
 <4\eta-\frac{B}{2m},
 \end{split}
\]
where the strict inequality follows from \(q=\delta/m\)
and \(\delta<1/4\).

By our choice of parameters,
\(B/(2m)\ge(99/200)\sigma/\sqrt{m\delta}\) and
\(4\eta\le\sigma/(25\sqrt{2m\delta})\).
Therefore, on \(A\),
\[
 |\widehat\Delta^{\mathsf{pool}}-\Delta|
 >\frac{B}{2m}-4\eta
 >\frac{\sigma}{2\sqrt{2m\delta}}
 >\frac14\sigma\sqrt{\frac{\log(2/\delta)}{m}},
\]
where the last inequality uses
\(\delta\log(2/\delta)<(\log8)/4<2\).
Together with \(\Prob(A)>\delta\), this proves the theorem
with \(L=1/4\).
\Halmos\endproof

\subsection{Proof of Theorem~\ref{thm:sep-winsor-mean}}
\label{app:adaptive-separate}

We first state the supporting lemmas and then prove
Theorem~\ref{thm:sep-winsor-mean}, followed by the proofs
of the supporting lemmas. The argument bounds the empirical
winsorized mean between two sample means evaluated at
population cutoffs.

\subsubsection{Supporting Lemmas}

We use the following form of Bernstein's inequality;
see \citet[Theorem~3.1.7, equation~(3.24)]
{gine2016mathematical}.

\begin{lemma}[Bernstein's inequality]\label{lem:bernstein}
Let \(Z_1,\ldots,Z_m\) be independent, mean-zero random
variables satisfying \(|Z_i|\le M\) almost surely, and
write \(v=\sum_{i=1}^m\E Z_i^2\). For every \(u>0\),
\[
 \Prob\left\{\sum_{i=1}^m Z_i>
 \sqrt{2vu}+\frac{Mu}{3}\right\}\le e^{-u}.
\]
\end{lemma}

Let \(X_1,\ldots,X_m\) be i.i.d.\ with continuous
distribution function \(F\), mean \(\mu\), and variance
at most \(\sigma^2\). For \(p\in(0,1)\), define
\(Q_p=\inf\{x:F(x)\ge p\}\).
For deterministic winsorization cutoffs \(a<b\), write
\[
 \phi_{a,b}(x)=\min\{b,\max\{a,x\}\},\qquad
 \widehat\mu_{a,b}=\frac1m\sum_{i=1}^m\phi_{a,b}(X_i),
 \qquad
 \mu_{a,b}=\E\phi_{a,b}(X_1).
\]
Fix \(c>1\) and \(0<\varepsilon\le1/4\), and define
\[
 A_c=2c^2+\frac{8c}{3},\qquad
 B_c=\sqrt{\frac{4(c+1)}{3c-1}},\qquad
 D_c=\sqrt{\frac{2}{1-c^{-2}}}.
\]
Let \(k=\lceil m\varepsilon\rceil\),
\(\widehat a=X_{(k)}\), and
\(\widehat b=X_{(m+1-k)}\).
Set \(p_-=\varepsilon(1-1/c)\) and
\(p_+=\varepsilon(1+1/c)\), and define
\[
 a^-=Q_{p_-},\qquad a^+=Q_{p_+},\qquad
 b^-=Q_{1-p_+},\qquad b^+=Q_{1-p_-}.
\]

Then, the following lemma locates both empirical cutoffs between
their corresponding population quantiles.

\begin{lemma}\label{lem:quantile-localization}
For \(u>0\), if \(m\varepsilon\ge A_cu\), then
\begin{equation}\label{eq:adaptive-localization}
 \Prob\{a^-\le\widehat a\le a^+,\,
        b^-\le\widehat b\le b^+\}
 \ge1-4e^{-u}.
\end{equation}
\end{lemma}

The next lemma bounds the population bias and the range
of each deterministic transformation in this inequality.

\begin{lemma}\label{lem:adaptive-tail-moments}
For each of the pairs \((a^-,b^-)\) and \((a^+,b^+)\),
\begin{equation}\label{eq:adaptive-population-bounds}
 |\mu_{a^\pm,b^\pm}-\mu|
 \le B_c\sigma\sqrt{\varepsilon},
 \qquad
 b^\pm-a^\pm\le\frac{D_c\sigma}{\sqrt{\varepsilon}}.
\end{equation}
\end{lemma}

Finally, the following lemma controls the sampling error
for deterministic winsorization cutoffs, which derives from a direct application of Bernstein's inequality from Lemma~\ref{lem:bernstein}.

\begin{lemma}\label{lem:winsorized-concentration}
For deterministic \(a<b\) and every \(u>0\),
\[
 \Prob\left\{
 \widehat\mu_{a,b}-\mu_{a,b}
 >\sigma\sqrt{\frac{2u}{m}}+\frac{(b-a)u}{3m}
 \right\}\le e^{-u}.
\]
The same inequality holds for
\(\mu_{a,b}-\widehat\mu_{a,b}\).
\end{lemma}

\subsubsection{Proof of Theorem~\ref{thm:sep-winsor-mean}}

\proof{Proof.}
Set \(u_\delta=\log(24/\delta)\) and
\(\varepsilon=A_cu_\delta/n_{\min}\le1/4\).
First consider one cell \((g,t)\), with sample size
\(n_g\ge n_{\min}\), and apply the preceding notation
with \(m=n_g\) and \(\mu=\mu_{gt}\).
Write
\(\widehat\mu_{gt,\varepsilon}^{\mathsf{sep}}
=\widehat\mu_{\widehat a,\widehat b}\)
for its empirical winsorized mean.

Using \(x_+=\max\{x,0\}\), we decompose the estimation
error into three contributions:
\[
 \begin{split}
 |\widehat\mu_{\widehat a,\widehat b}-\mu|
 \le{}&
 \underbrace{
 \max\left\{
 (\widehat\mu_{\widehat a,\widehat b}
       -\widehat\mu_{a^+,b^+})_+,\,
 (\widehat\mu_{a^-,b^-}
       -\widehat\mu_{\widehat a,\widehat b})_+
 \right\}}_{D_1}\\
 &+
 \underbrace{
 \max\left\{
 \widehat\mu_{a^+,b^+}-\mu_{a^+,b^+},\,
 \mu_{a^-,b^-}-\widehat\mu_{a^-,b^-},\,0
 \right\}}_{D_2}\\
 &+
 \underbrace{
 \max\left\{
 |\mu_{a^+,b^+}-\mu|,\,
 |\mu_{a^-,b^-}-\mu|
 \right\}}_{D_3}.
 \end{split}
\]
Here, \(D_1\) measures how far the empirical winsorized
mean falls outside the two sample means evaluated at
population cutoffs; \(D_2\) represents their sampling
error in the relevant directions; and \(D_3\) represents
their population winsorization bias.
The displayed inequality follows by adding and subtracting
\(\widehat\mu_{a^+,b^+}\) and \(\mu_{a^+,b^+}\)
for the upper error, and
\(\widehat\mu_{a^-,b^-}\) and \(\mu_{a^-,b^-}\)
for the lower error.

Define the cutoff-localization event
\[
 \mathcal E_1=
 \{a^-\le\widehat a\le a^+,\,
   b^-\le\widehat b\le b^+\}.
\]
Since \(n_g\varepsilon\ge A_cu_\delta\),
Lemma~\ref{lem:quantile-localization} gives
\(\Prob(\mathcal E_1)\ge1-4e^{-u_\delta}\).
On \(\mathcal E_1\), monotonicity in both cutoffs gives
\(\widehat\mu_{a^-,b^-}
 \le\widehat\mu_{\widehat a,\widehat b}
 \le\widehat\mu_{a^+,b^+}\), so \(D_1=0\).

Next, set
\[
 r_g=\sigma\sqrt{\frac{2u_\delta}{n_g}}
      +\frac{D_c\sigma u_\delta}
             {3n_g\sqrt{\varepsilon}},
 \qquad
 \mathcal E_2=
 \left\{
 \widehat\mu_{a^+,b^+}-\mu_{a^+,b^+}\le r_g,\,
 \mu_{a^-,b^-}-\widehat\mu_{a^-,b^-}\le r_g
 \right\}.
\]
Lemma~\ref{lem:winsorized-concentration}, together with
the range bound \(b^\pm-a^\pm\le
D_c\sigma/\sqrt{\varepsilon}\) from
Lemma~\ref{lem:adaptive-tail-moments}, gives
\(\Prob(\mathcal E_2)\ge1-2e^{-u_\delta}\).
Thus \(D_2\le r_g\) on \(\mathcal E_2\).
Finally, the bias bound in
Lemma~\ref{lem:adaptive-tail-moments} gives
\(D_3\le B_c\sigma\sqrt{\varepsilon}\) deterministically.

Therefore, on \(\mathcal E_1\cap\mathcal E_2\),
whose probability is at least
\(1-6e^{-u_\delta}=1-\delta/4\),
\[
 \begin{split}
 |\widehat\mu_{gt,\varepsilon}^{\mathsf{sep}}-\mu_{gt}|
 &\le D_1+D_2+D_3\\
 &\le B_c\sigma\sqrt{\varepsilon}
      +\sigma\sqrt{\frac{2u_\delta}{n_g}}
      +\frac{D_c\sigma u_\delta}
             {3n_g\sqrt{\varepsilon}}.
 \end{split}
\]
Substituting \(\varepsilon=A_cu_\delta/n_{\min}\)
and using \(n_g\ge n_{\min}\), we obtain
\begin{equation}\label{eq:one-cell-bound}
 |\widehat\mu_{gt,\varepsilon}^{\mathsf{sep}}-\mu_{gt}|
 \le
 \left(B_c\sqrt{A_c}+\sqrt2+\frac{D_c}{3\sqrt{A_c}}\right)
 \sigma\sqrt{\frac{u_\delta}{n_{\min}}}.
\end{equation}
Define
\begin{equation}\label{eq:appendix-adaptive-constants}
 C_c=B_c\sqrt{A_c}+\sqrt2+\frac{D_c}{3\sqrt{A_c}},
 \qquad L(c)=4C_c.
\end{equation}

A union bound over the four cells gives an event of
probability at least \(1-\delta\) on which
\eqref{eq:one-cell-bound} holds simultaneously.
On this event,
\[
 \begin{split}
 |\widehat\Delta^{\mathsf{sep}}-\Delta|
 &\le
 |\widehat\mu_{11,\varepsilon}^{\mathsf{sep}}-\mu_{11}|
 +|\widehat\mu_{10,\varepsilon}^{\mathsf{sep}}-\mu_{10}|\\
 &\quad+
 |\widehat\mu_{01,\varepsilon}^{\mathsf{sep}}-\mu_{01}|
 +|\widehat\mu_{00,\varepsilon}^{\mathsf{sep}}-\mu_{00}|\\
 &\le L(c)\sigma
       \sqrt{\frac{\log(24/\delta)}{n_{\min}}}.
 \end{split}
\]
This proves \eqref{eq:sep-winsor-mean-bound} with
\(L=L(c)\).
\Halmos\endproof

\subsubsection{Proof of Lemma~\ref{lem:quantile-localization}}

\proof{Proof.}
Recall that \(k=\lceil m\varepsilon\rceil\),
\(\widehat a=X_{(k)}\), and
\(\widehat b=X_{(m+1-k)}\).
The population cutoffs are
\[
 a^-=Q_{p_-},\qquad a^+=Q_{p_+},\qquad
 b^-=Q_{1-p_+},\qquad b^+=Q_{1-p_-},
\]
where \(p_-=\varepsilon(1-1/c)\) and
\(p_+=\varepsilon(1+1/c)\).
Define the four failure events
\[
 E_1=\{\widehat a<a^-\},\quad
 E_2=\{\widehat a>a^+\},\quad
 E_3=\{\widehat b<b^-\},\quad
 E_4=\{\widehat b>b^+\}.
\]
By continuity,
\(F(a^-)=p_-\), \(F(a^+)=p_+\),
\(\Prob(X>b^-)=p_+\), and
\(\Prob(X>b^+)=p_-\).
Almost surely, no observation equals any of these
deterministic thresholds.

On \(E_1\), the \(k\)-th order statistic lies below
\(a^-=Q_{p_-}\), so at least \(k\) observations lie
below this population quantile. Since \(k/m\ge\varepsilon\),
their empirical fraction exceeds its expectation
\(p_-=\varepsilon(1-1/c)\) by at least
\(\varepsilon/c\):
\[
 \frac1m\sum_{i=1}^m\ind\{X_i<a^-\}-p_-
 \ge\frac{k}{m}-\varepsilon(1-1/c)
 \ge\frac{\varepsilon}{c}.
\]
On \(E_2\), the \(k\)-th order statistic lies above
\(a^+=Q_{p_+}\), so at most \(k-1\) observations lie
below this quantile. Because \((k-1)/m<\varepsilon\),
their empirical fraction falls below its expectation
\(p_+=\varepsilon(1+1/c)\) by more than
\(\varepsilon/c\):
\[
 p_+-\frac1m\sum_{i=1}^m\ind\{X_i<a^+\}
 \ge\varepsilon(1+1/c)-\frac{k-1}{m}
 >\frac{\varepsilon}{c}.
\]

For the upper cutoff, on \(E_3\), the
\((m+1-k)\)-th order statistic lies below
\(b^-=Q_{1-p_+}\). Thus at least \(m+1-k\)
observations lie below \(b^-\), leaving at most
\(k-1\) observations above it. Their empirical
upper-tail fraction therefore satisfies
\[
 p_+-\frac1m\sum_{i=1}^m\ind\{X_i>b^-\}
 \ge\varepsilon(1+1/c)-\frac{k-1}{m}
 >\frac{\varepsilon}{c}.
\]
On \(E_4\), the observations with ranks
\(m+1-k,\ldots,m\) all lie above
\(b^+=Q_{1-p_-}\). There are \(k\) such observations,
so their empirical upper-tail fraction satisfies
\[
 \frac1m\sum_{i=1}^m\ind\{X_i>b^+\}-p_-
 \ge\frac{k}{m}-\varepsilon(1-1/c)
 \ge\frac{\varepsilon}{c}.
\]
Consequently, each \(E_j\) requires a deviation of
at least \(\varepsilon/c\) in one specified direction
from the corresponding population tail probability. Then, we will apply the Bernstein inequality to show that each of these events occurs with probability $e^{-u}$. Specifically, since each empirical fraction is an average of independent
Bernoulli variables denoted as $I_i$ with success probability
\(p\in\{p_-,p_+\}\), 
applying Lemma~\ref{lem:bernstein} with
\(u_0=m\varepsilon/A_c\), we obtain
\[
 \Prob\left\{
 \frac1m\sum_{i=1}^m I_i-p
 >\sqrt{\frac{2pu_0}{m}}+\frac{u_0}{3m}
 \right\}\le e^{-u_0},
\]
and the same bound holds for the lower deviation.
Since \(p\le p_+=\varepsilon(1+1/c)\) and
\(A_c=2c^2+8c/3\), the threshold satisfies
\[
 \sqrt{\frac{2pu_0}{m}}+\frac{u_0}{3m}
 \le\varepsilon\left\{
 \sqrt{\frac{2(1+1/c)}{A_c}}+\frac{1}{3A_c}
 \right\}
 <\frac{\varepsilon}{c}.
\]
Each \(E_j\) requires a deviation of at least
\(\varepsilon/c\) in one specified direction. Hence,
\[
 \Prob(E_j)\le e^{-u_0}
 =\exp\left\{-\frac{m\varepsilon}{A_c}\right\}
 \le e^{-u},
 \qquad j=1,\ldots,4,
\]
where the last inequality uses \(m\varepsilon\ge A_cu\).
A union bound therefore gives
\[
 \Prob\{a^-\le\widehat a\le a^+,\,
        b^-\le\widehat b\le b^+\}
 =1-\Prob\left(\bigcup_{j=1}^4 E_j\right)
 \ge1-4e^{-u},
\]
which proves \eqref{eq:adaptive-localization}.
\Halmos\endproof

\subsubsection{Proof of Lemma~\ref{lem:adaptive-tail-moments}}

\proof{Proof.}
For \(p\in(0,1)\),
let \(\mu_-=\E(X\mid X<Q_p)\) and
\(\mu_+=\E(X\mid X\ge Q_p)\).
Continuity of the distribution function \(F\) implies
\(\Prob(X<Q_p)=p\), and \(\mu_-\le Q_p\le\mu_+\).
Let \(I=\ind\{X<Q_p\}\). By the law of total variance,
\[
 \begin{split}
 \Var(X)
 &=\E\{\Var(X\mid I)\}
   +\Var\{\E(X\mid I)\}\\
 &=\E\{\Var(X\mid I)\}
 +p(\mu_--\mu)^2+(1-p)(\mu_+-\mu)^2.
 \end{split}
\]
Since \(\mu=p\mu_-+(1-p)\mu_+\), we have
\(\mu_--\mu=-(1-p)(\mu_+-\mu_-)\) and
\(\mu_+-\mu=p(\mu_+-\mu_-)\). Hence,
\[
 p(\mu_--\mu)^2+(1-p)(\mu_+-\mu)^2
 =p(1-p)(\mu_+-\mu_-)^2.
\]
The conditional variances are nonnegative, so
\[
 \sigma^2\ge\Var(X)
 \ge p(1-p)(\mu_+-\mu_-)^2.
\]
Recall that \(x_+=\max\{x,0\}\). Since
\((Q_p-X)_+=(Q_p-X)\ind\{X<Q_p\}\),
conditioning on the lower-tail event gives
\[
 \begin{split}
 \E(Q_p-X)_+
 &=\Prob(X<Q_p)\E(Q_p-X\mid X<Q_p)\\
 &=p\{Q_p-\E(X\mid X<Q_p)\}\\
 &=p(Q_p-\mu_-).
 \end{split}
\]
Moreover, since \(\mu_+=\E(X\mid X\ge Q_p)\ge Q_p\),
we have
\(p(Q_p-\mu_-)\le p(\mu_+-\mu_-)\).
The variance bound established above gives
\(p(1-p)(\mu_+-\mu_-)^2\le\sigma^2\).
Since \(\mu_+-\mu_-\ge0\), taking square roots yields
\(\mu_+-\mu_-\le\sigma/\sqrt{p(1-p)}\).
Combining these observations, we obtain
\[
 \E(Q_p-X)_+
 =p(Q_p-\mu_-)
 \le p(\mu_+-\mu_-)
 \le\frac{p\sigma}{\sqrt{p(1-p)}}
 =\sigma\sqrt{\frac{p}{1-p}}.
\]

For deterministic quantiles \(a=Q_{p_L}<b=Q_{1-p_U}\),
the identity
\(\phi_{a,b}(X)-X=(a-X)_+-(X-b)_+\) gives
\(\mu_{a,b}-\mu=\E(a-X)_+-\E(X-b)_+\).
Since both expectations are nonnegative,
the absolute value of their difference is at most their
maximum. Applying the individual tail bounds gives
\[
 |\mu_{a,b}-\mu|
 \le\max\{\E(a-X)_+,\E(X-b)_+\}
 \le\sigma\max\left\{
 \sqrt{\frac{p_L}{1-p_L}},
 \sqrt{\frac{p_U}{1-p_U}}
 \right\}.
\]
For \((a^-,b^-)\), we have
\((p_L,p_U)=(p_-,p_+)\); for \((a^+,b^+)\),
we have \((p_L,p_U)=(p_+,p_-)\).
Because \(p_-<p_+\) and \(p/(1-p)\) is increasing,
both pairs satisfy
\[
 \begin{split}
 |\mu_{a^\pm,b^\pm}-\mu|
 &\le\sigma\sqrt{\frac{p_+}{1-p_+}}
 =\sigma\sqrt{\varepsilon}
   \sqrt{\frac{c+1}{c-\varepsilon(c+1)}}\\
 &\le\sigma\sqrt{\varepsilon}
   \sqrt{\frac{4(c+1)}{3c-1}}
 =B_c\sigma\sqrt{\varepsilon},
 \end{split}
\]
where \(\varepsilon\le1/4\) implies
\(c-\varepsilon(c+1)\ge(3c-1)/4>0\).

To bound the range, let \(m_L=\E(X\mid X<a)\)
and \(m_U=\E(X\mid X>b)\), so
\(m_L\le a<b\le m_U\).
By definition of the quantile we have \(\Prob(X<a)=p_L\) and
\(\Prob(X>b)=p_U\). Applying a similar conditional variance inequality argument yields
\[
 \sigma^2\ge\Var(X)
 \ge p_L(m_L-\mu)^2+p_U(m_U-\mu)^2.
\]
For \(\bar m=(p_Lm_L+p_Um_U)/(p_L+p_U)\),
completing the square gives
\[
 \begin{split}
 p_L(m_L-\mu)^2+p_U(m_U-\mu)^2
 &=\frac{p_Lp_U}{p_L+p_U}(m_U-m_L)^2
   +(p_L+p_U)(\bar m-\mu)^2\\
 &\ge\frac{p_Lp_U}{p_L+p_U}(b-a)^2,
 \end{split}
\]
where the inequality uses \(m_U-m_L\ge b-a\).
Thus \(b-a\le\sigma\sqrt{1/p_L+1/p_U}\).
For either pair, \(\{p_L,p_U\}=\{p_-,p_+\}\), so
\[
 b^\pm-a^\pm
 \le\sigma\sqrt{\frac1{p_-}+\frac1{p_+}}
 =\frac{\sigma}{\sqrt{\varepsilon}}
   \sqrt{\frac1{1-1/c}+\frac1{1+1/c}}
 =\frac{D_c\sigma}{\sqrt{\varepsilon}}.
\]
This proves both assertions.
\Halmos\endproof

\subsubsection{Proof of Lemma~\ref{lem:winsorized-concentration}}

\proof{Proof.}
Let \(Z_i=\phi_{a,b}(X_i)-\mu_{a,b}\) and \(R=b-a\).
Then the \(Z_i\) are independent and mean zero,
with \(|Z_i|\le R\).
For an independent copy \(X'\) of \(X\), the
one-Lipschitz property of \(\phi_{a,b}\) gives
\[
 \begin{split}
 \Var\{\phi_{a,b}(X)\}
 &=\frac12\E\{\phi_{a,b}(X)-\phi_{a,b}(X')\}^2\\
 &\le\frac12\E(X-X')^2
 =\Var(X)\le\sigma^2.
 \end{split}
\]
Hence \(v=\sum_{i=1}^m\E Z_i^2\le m\sigma^2\).
Applying Lemma~\ref{lem:bernstein} yields
\[
 \Prob\left\{\sum_{i=1}^m Z_i>
 \sqrt{2m\sigma^2u}+\frac{Ru}{3}\right\}
 \le e^{-u}.
\]
Dividing the threshold by \(m\) proves the
upper-deviation bound. Applying the same inequality
to \(-Z_i\) proves the lower-deviation bound.
\Halmos\endproof

\section{Proofs for Section~\ref{sec:asymptotics}}
\label{app:asymptotics}

We first establish the bias decomposition and the associated
unbiasedness conditions. We then prove
Theorem~\ref{thm:sep-fixed-clt} and the subsequent efficiency
comparisons.

\subsection{Bias of Separate and Pooled Winsorization}

Recall the location-shift model
$Y_{gt,i}=\mu_{gt}+e_{gt,i}$ in
\eqref{eq:location-shift-asymptotic}.
The errors are mutually independent across individuals,
groups, and periods, with
$\E(e_{gt,i})=0$ and distributions in $\cP(\sigma^2)$.
Within each cell $(g,t)$, they are identically distributed
with distribution function $F_{gt}$, which may differ
across cells. We also assume that each $F_{gt}$ is
continuous and strictly increasing. Throughout this
subsection, fix $\eps\in(0,1/2)$.

\proof{The bias decomposition.}
For any integrable random variable $Y$ and deterministic winsorization
cutoffs $A<B$, the definition of winsorization gives
$\phi_{A,B}(Y)-Y
=(A-Y)\ind\{Y<A\}-(Y-B)\ind\{Y>B\}$.
Applying this identity to a generic error
$e\sim F_{gt}$ and using $\E(e)=0$ proves
\eqref{eq:beta-definition}.

Moreover, since shifting both the observation and its winsorization
cutoffs by the same amount shifts the winsorized value
by that amount, we have
$\phi_{A,B}(\mu_{gt}+e)
=\mu_{gt}+\phi_{A-\mu_{gt},B-\mu_{gt}}(e)$.
Taking expectations proves
\eqref{eq:location-winsor-mean}.
For either method
$j\in\{\mathsf{sep},\mathsf{pool}\}$, substituting the
four cell expectations into the DiD formula yields
\[
B^j_{\eps} = \Delta_\eps^j-\Delta
=(\beta_{11}^j-\beta_{10}^j)
 -(\beta_{01}^j-\beta_{00}^j),
\]
which is \eqref{eq:unified-DiD-bias}.
Consequently, the population-cutoff contrast equals
$\Delta$ if and only if
$\beta_{11}^j-\beta_{10}^j
=\beta_{01}^j-\beta_{00}^j$.
\Halmos\endproof

\proof{Proof of Proposition~\ref{prop:sep-unbiasedness}.}
The lower and upper outcome quantiles in cell $(g,t)$
are $\mu_{gt}+\ell_{gt,\eps}$ and
$\mu_{gt}+u_{gt,\eps}$, respectively.
As stated in \eqref{eq:sep-cell-bias}, the cellwise bias reduces to $\beta_{gt}^{\mathsf{sep}}
=\E\{\phi_{\ell_{gt,\eps},u_{gt,\eps}}(e)\}$. To check whether any of the three assumptions is sufficient for $B^{\mathsf{sep}}_{\eps} = 0$,
we verify the three conditions separately.

Under (\ref{cond:sep-symmetry}), the mean-zero error
distribution is symmetric about zero, so
$\ell_{gt,\eps}=-u_{gt,\eps}$.
The map $\phi_{-u_{gt,\eps},u_{gt,\eps}}$ is odd and
bounded. Symmetry therefore implies
$\beta_{gt}^{\mathsf{sep}}
=\E\{\phi_{-u_{gt,\eps},u_{gt,\eps}}(e)\}=0$.
Thus all four cell biases vanish.

Under (\ref{cond:sep-time-stable}), the error
distributions in the two periods are identical within
each group. Their quantiles and winsorized expectations
are consequently identical, giving
$\beta_{g1}^{\mathsf{sep}}
=\beta_{g0}^{\mathsf{sep}}$ for each $g$.
Under (\ref{cond:sep-group-stable}), the same argument
across groups within each period gives
$\beta_{1t}^{\mathsf{sep}}
=\beta_{0t}^{\mathsf{sep}}$ for each $t$.
In every case, \eqref{eq:unified-DiD-bias} implies
$B_\eps^{\mathsf{sep}}=0$.
\Halmos\endproof

\proof{Proof of Proposition~\ref{prop:pool-normal-bias}.}
The common Gaussian error distribution satisfies all
three conditions in
Proposition~\ref{prop:sep-unbiasedness}.
Hence $\Delta_\eps^{\mathsf{sep}}=\Delta=d$.

Let $Z\sim N(0,\sigma^2)$, and let $a_P<b_P$ be the
pooled population cutoffs. These are finite
deterministic quantities for the given model and
fixed $\eps$. Since the other three cells have
the same distribution as $Z$, we obtain
\[
\Delta_\eps^{\mathsf{pool}}
=\E\{\phi_{a_P,b_P}(Z+d)-\phi_{a_P,b_P}(Z)\}.
\]
The winsorization map is nondecreasing and
one-Lipschitz, so the difference inside the expectation
belongs to $[0,d]$. It is strictly positive whenever
$a_P<Z<b_P$ and equals zero whenever $Z\ge b_P$.
Both events have positive probability under the
nondegenerate Gaussian distribution. Therefore
$0<\Delta_\eps^{\mathsf{pool}}<d$, proving the claim.
\Halmos\endproof

\subsection{Proof of Theorem~\ref{thm:sep-fixed-clt}}
\label{app:empirical-cutoff-clt}

\proof{Proof.}
Recall the location-shift model
\eqref{eq:location-shift-asymptotic}, under which
the errors are mutually independent across individuals,
groups, and periods and are identically distributed
within each cell. Fix $\eps\in(0,1/2)$.
For each cell, the density is continuous and strictly
positive in neighborhoods of the two population
quantiles, as assumed before
Theorem~\ref{thm:sep-fixed-clt}.
We first derive an expansion for a single cell and
then combine the four cell estimators.

\medskip
\noindent\textit{Expansion of the empirical cutoffs.}
Fix a cell and suppress its indices.
Write $m=n_g$, $F=F_{gt}$, $f=f_{gt}$,
$\ell=F^{-1}(\eps)$, and $u=F^{-1}(1-\eps)$.
For its error observations $e_1,\ldots,e_m$, define
$\widehat F_m(x)=m^{-1}\sum_{i=1}^m\ind\{e_i\le x\}$.
Let $k_m=\lceil m\eps\rceil$,
$\widehat\ell=e_{(k_m)}$, and
$\widehat u=e_{(m+1-k_m)}$.

We use \citet[Theorem~1, p.~1957]{ghosh1971new}.
In the notation needed here, that theorem states
that, for an i.i.d.\ sample from $F$, if
$F(q)=p\in(0,1)$, $F$ is differentiable at $q$
with $F'(q)>0$, and a deterministic sequence
$p_m$ satisfies $p_m-p=O(m^{-1/2})$, then a
sample $p_m$-quantile $\widehat q_{p_m}$ satisfies
\[
\widehat q_{p_m}-q
=
\frac{p_m-\widehat F_m(q)}{F'(q)}
+o_p(m^{-1/2}).
\]
This is the representation in equation~(2) of
Ghosh's paper, with the remainder controlled by
his Theorem~1.

For the lower cutoff, take $p=\eps$ and
$p_m=k_m/m$; for the upper cutoff, take
$p=1-\eps$ and $p_m=(m+1-k_m)/m$.
Both empirical quantile levels $p_m$ differ from their limits $p$ by
$O(m^{-1})$. The density assumptions ensure that
$F'(\ell)=f(\ell)>0$ and $F'(u)=f(u)>0$.
Thus all the conditions of the cited theorem hold,
and
\begin{equation}
\label{eq:app-quantile-representation}
\begin{split}
\widehat\ell-\ell
&=
\frac{\eps-\widehat F_m(\ell)}{f(\ell)}
+o_p(m^{-1/2}),\\
\widehat u-u
&=
\frac{1-\eps-\widehat F_m(u)}{f(u)}
+o_p(m^{-1/2}).
\end{split}
\end{equation}
Since $F(\ell)=\eps$ and $F(u)=1-\eps$, the
indicators defining the empirical distribution
function at these quantiles are i.i.d.\ Bernoulli
variables. In particular,
$\Var\{\widehat F_m(\ell)\}
=\Var\{\widehat F_m(u)\}
=\eps(1-\eps)/m$.
Chebyshev's inequality therefore gives
$\widehat F_m(\ell)-\eps=O_p(m^{-1/2})$ and
$\widehat F_m(u)-(1-\eps)=O_p(m^{-1/2})$.
Combining these bounds with
\eqref{eq:app-quantile-representation} and
$f(\ell),f(u)>0$ yields
$\widehat\ell-\ell=O_p(m^{-1/2})$ and
$\widehat u-u=O_p(m^{-1/2})$.

\medskip
\noindent\textit{Expansion of the winsorized cell mean.}
Define $\widehat T_\eps=m^{-1}\sum_{i=1}^m
\phi_{\widehat\ell,\widehat u}(e_i)$ and
$\beta=\E\{\phi_{\ell,u}(e)\}$. Adding and subtracting the
sample mean with population cutoffs gives
\begin{equation}\label{eq:app-win-two-parts}
\begin{split}
\widehat T_\eps-\beta
={}&\frac1m\sum_{i=1}^m\{\phi_{\ell,u}(e_i)-\beta\}+\frac1m\sum_{i=1}^m
\{\phi_{\widehat\ell,\widehat u}(e_i)-\phi_{\ell,u}(e_i)\}.
\end{split}
\end{equation}
The second term satisfies the exact identity
\[
\frac1m\sum_{i=1}^m
\{\phi_{\widehat\ell,\widehat u}(e_i)-\phi_{\ell,u}(e_i)\}
=
\int_\ell^{\widehat\ell}\widehat F_m(t)\,dt
+\int_u^{\widehat u}\{1-\widehat F_m(t)\}\,dt,
\]
which follows by integrating
$\ind\{e_i\le t\}$ and $\ind\{e_i>t\}$ over the respective
cutoff intervals. Integrals with reversed limits are signed,
and the identity requires no independence between the cutoffs
and the observations.

Since $F(\ell)=\eps$ and $F(u)=1-\eps$, the empirical
distribution function at either fixed cutoff is an average
of independent Bernoulli variables with variance
$\eps(1-\eps)/m$. Chebyshev's inequality therefore gives
$\widehat F_m(\ell)-\eps=O_p(m^{-1/2})$ and
$\widehat F_m(u)-(1-\eps)=O_p(m^{-1/2})$.
Moreover, continuity of $F$ excludes ties almost surely, so
$\widehat F_m(\widehat\ell)=k_m/m$ and
$1-\widehat F_m(\widehat u)=(k_m-1)/m$, where
$k_m=\lceil m\eps\rceil$; both values differ from $\eps$
by at most $1/m$. Monotonicity of $\widehat F_m$, together
with the previously established cutoff rates for $\hat{\ell}$ and $\hat{u}$ in \eqref{eq:app-quantile-representation}, now yields
\[
\begin{aligned}
\left|\int_\ell^{\widehat\ell}
\{\widehat F_m(t)-\eps\}\,dt\right|
&\le |\widehat\ell-\ell|
\max\left\{
|\widehat F_m(\ell)-\eps|,
\left|\frac{k_m}{m}-\eps\right|
\right\}
=O_p(m^{-1}),\\
\left|\int_u^{\widehat u}
\{1-\widehat F_m(t)-\eps\}\,dt\right|
&\le |\widehat u-u|
\max\left\{
|1-\widehat F_m(u)-\eps|,
\left|\frac{k_m-1}{m}-\eps\right|
\right\}
=O_p(m^{-1}).
\end{aligned}
\]
Substituting these bounds into
\eqref{eq:app-win-two-parts}, we obtain
\begin{equation}\label{eq:app-win-cutoff-expansion}
\begin{split}
\widehat T_\eps-\beta
={}&\frac1m\sum_{i=1}^m\{\phi_{\ell,u}(e_i)-\beta\}\\
&+\eps(\widehat\ell-\ell)+\eps(\widehat u-u)+O_p(m^{-1}).
\end{split}
\end{equation}
The quantile representation in
\eqref{eq:app-quantile-representation} then gives
\begin{equation}\label{eq:app-cell-linear}
\widehat T_\eps-\beta
=\frac1m\sum_{i=1}^m\psi_{F,\eps}(e_i)+o_p(m^{-1/2}),
\end{equation}
where
\[
\psi_{F,\eps}(e)
=\phi_{\ell,u}(e)-\beta
+\eps\frac{\eps-\ind\{e\le\ell\}}{f(\ell)}
+\eps\frac{1-\eps-\ind\{e\le u\}}{f(u)}.
\]
This is the influence function in \eqref{eq:cell-if-general}.
It has mean zero and is bounded under the stated cutoff
conditions. Hence, the ordinary central limit theorem and
\eqref{eq:app-cell-linear} imply
$\sqrt m(\widehat T_\eps-\beta)
\Rightarrow N(0,\Var\{\psi_{F,\eps}(e)\})$.

\medskip
\noindent\textit{Combining the four cells.}
Restore the cell indices and write
$\beta_{gt}=\E\{\phi_{\ell_{gt,\eps},u_{gt,\eps}}(e_{gt,i})\}$.
Let $\widehat\ell_{gt,\eps}$ and $\widehat u_{gt,\eps}$
denote the empirical error quantiles. Under the
location-shift model $Y_{gt,i}=\mu_{gt}+e_{gt,i}$,
adding $\mu_{gt}$ preserves the ordering of the observations,
so $Y_{gt,(j)}=\mu_{gt}+e_{gt,(j)}$ for every order statistic.
Consequently, the empirical outcome cutoffs are
$\mu_{gt}+\widehat\ell_{gt,\eps}$ and
$\mu_{gt}+\widehat u_{gt,\eps}$.

Moreover, the definition
$\phi_{a,b}(x)=\min\{b,\max\{a,x\}\}$ gives, for every $c$,
\[
\phi_{a+c,b+c}(x+c)
=\min\{b+c,\max\{a+c,x+c\}\}
=c+\phi_{a,b}(x).
\]
Applying this identity observation by observation yields
\[
\begin{split}
&\frac1{n_g}\sum_{i=1}^{n_g}
\phi_{\mu_{gt}+\widehat\ell_{gt,\eps},
      \mu_{gt}+\widehat u_{gt,\eps}}(Y_{gt,i})
-\mu_{gt}-\beta_{gt}\\
&\qquad=
\frac1{n_g}\sum_{i=1}^{n_g}
\phi_{\widehat\ell_{gt,\eps},
      \widehat u_{gt,\eps}}(e_{gt,i})-\beta_{gt}\\
&\qquad=
\frac1{n_g}\sum_{i=1}^{n_g}
\psi_{gt,\eps}(e_{gt,i})+o_p(n_g^{-1/2}),
\end{split}
\]
where the first equality is exact and the second follows
from the single-cell expansion
\eqref{eq:app-cell-linear} with $m=n_g$.
Under any one of
(\ref{cond:sep-symmetry})--(\ref{cond:sep-group-stable}),
Proposition~\ref{prop:sep-unbiasedness} gives
$(\beta_{11}-\beta_{10})-(\beta_{01}-\beta_{00})=0$.
Thus the population cell biases cancel.

Define $G_{gt,n}=n_g^{-1/2}\sum_{i=1}^{n_g}
\psi_{gt,\eps}(e_{gt,i})$ and
$v_{gt,\eps}=\Var\{\psi_{gt,\eps}(e_{gt,i})\}$.
Independence across cells and the cellwise central limit
theorems imply joint convergence of the four $G_{gt,n}$
to independent normal variables with variances
$v_{gt,\eps}$. Combining the cell expansions gives
\[
\sqrt n(\widehat\Delta_\eps^{\mathsf{sep}}-\Delta)
=
\sqrt{\frac{n}{n_1}}(G_{11,n}-G_{10,n})
-\sqrt{\frac{n}{n_0}}(G_{01,n}-G_{00,n})
+o_p(1),
\]
where $n_g/n\to p_g\in(0,1)$ ensures that the scaled
remainders are $o_p(1)$. Slutsky's theorem therefore yields
\[
\sqrt n(\widehat\Delta_\eps^{\mathsf{sep}}-\Delta)
\ \Rightarrow\
N\!\left(
0,\frac{v_{11,\eps}+v_{10,\eps}}{p_1}
+\frac{v_{01,\eps}+v_{00,\eps}}{p_0}
\right),
\]
whose limiting variance is $V_{\mathsf{sep}}(\eps)$
in \eqref{eq:vsep-general}.
\Halmos\endproof

\subsection{Additional Efficiency Comparisons by Tail Class}
\label{app:tail-class-proof}
\label{app:variance-identities}

We supplement the main-text DFR result with comparisons based on
two standard tail classes. Following the moment-generating-function
definition \citep{kahane1960proprietes}, a mean-zero distribution
belongs to $\mathcal P_{\mathrm{SG}}$ if some $K\in(0,\infty)$
satisfies $\E\exp(se)\le\exp(K^2s^2/2)$ for every $s\in\R$.
We use $\mathcal P_{\mathrm{HT}}$ for distributions satisfying
$\E\exp(s|e|)=\infty$ for every $s>0$, applying the classical
heavy-tail definition to $|e|$ \citep[Chapter~2]{foss2013heavy}.
Throughout this subsection, retain the independent sampling model
\eqref{eq:location-shift-asymptotic}, mean-zero cell errors in
$\cP(\sigma^2)$, and continuous, strictly increasing cell
distribution functions. The cutoff-density conditions used in
Theorem~\ref{thm:sep-fixed-clt} are required for every
$\eps\in(0,1/2)$.

\begin{theorem}\label{thm:tail-class-universal}
Under the preceding sampling and cutoff conditions, suppose
any one of (\ref{cond:sep-symmetry}),
(\ref{cond:sep-time-stable}), or
(\ref{cond:sep-group-stable}) holds.
\begin{enumerate}[label=\textup{(\roman*)},leftmargin=*,itemsep=4pt]
\item If every cell distribution belongs to
$\mathcal P_{\mathrm{SG}}$, then for every
$\eps_0\in(0,1/2)$ there exists $\eps\in(0,\eps_0)$
such that $V_{\mathsf{raw}}<V_{\mathsf{sep}}(\eps)$.
\item If all four cells share one common symmetric distribution
in $\mathcal P_{\mathrm{HT}}$, then for every
$\eps_0\in(0,1/2)$ there exists $\eps\in(0,\eps_0)$
such that $V_{\mathsf{sep}}(\eps)<V_{\mathsf{raw}}$.
\end{enumerate}
\end{theorem}

Both parts assert favorable fractions arbitrarily close to zero,
with a fraction chosen separately for each distribution and each
$\eps_0$. Neither part asserts an ordering at every sufficiently
small fraction. Part (i) permits heterogeneous cell distributions;
part (ii) additionally requires a common symmetric distribution.
The counterexamples below show why neither conclusion extends to
an ordering at every fraction throughout its tail class.

\subsubsection{Variance Reformulation and Proof}
\proof{Proof.}
By Proposition~\ref{prop:sep-unbiasedness}, separate
winsorization is asymptotically unbiased under any one
of conditions (\ref{cond:sep-symmetry}),
(\ref{cond:sep-time-stable}), or
(\ref{cond:sep-group-stable}).
We therefore compare its asymptotic variance in
Theorem~\ref{thm:sep-fixed-clt} with the raw DiD
variance in Section~\ref{sec:asymptotic-efficiency}.

Fix a cell and suppress its indices. Let $e$ have
distribution function $F$ and density $f$.
For $\eps\in(0,1/2)$, write
$a=F^{-1}(\eps)$, $b=F^{-1}(1-\eps)$,
$r_L=\eps/f(a)$, and $r_R=\eps/f(b)$.
The cutoff assumptions give $f(a),f(b)>0$ and
$\Prob(e<a)=\Prob(e>b)=\eps$.
Recall $\beta=\E\{\phi_{a,b}(e)\}$ from
\eqref{eq:beta-definition}, and define
\begin{equation}\label{eq:auxiliary-H}
H_{F,\eps}(x)=
\begin{cases}
a-r_L,&x<a,\\
x,&a\le x\le b,\\
b+r_R,&x>b,
\end{cases}
\qquad
\eta_{F,\eps}=\beta+\eps(r_R-r_L).
\end{equation}
Substitution into \eqref{eq:cell-if-general} gives
$\psi_{F,\eps}(e)=H_{F,\eps}(e)-\eta_{F,\eps}$
almost surely. Moreover,
\[
\E\{H_{F,\eps}(e)\}
=\E\{\phi_{a,b}(e)\}-\eps r_L+\eps r_R
=\eta_{F,\eps}.
\]
Thus $\Var\{\psi_{F,\eps}(e)\}
=\Var\{H_{F,\eps}(e)\}$.
Define
$\Gamma_F(\eps)=\Var(e)-\Var\{\psi_{F,\eps}(e)\}$.
Since $\E e=0$ and $H_{F,\eps}(e)=e$ on $[a,b]$,
\begin{equation}\label{eq:general-efficiency-index}
\begin{split}
\Gamma_F(\eps)
&=\E(e^2)-\E\{H_{F,\eps}(e)^2\}
+\eta_{F,\eps}^2\\
&=\eps\{\E(e^2\mid e<a)-(a-r_L)^2\}\\
&\quad+\eps\{\E(e^2\mid e>b)-(b+r_R)^2\}
+\eta_{F,\eps}^2.
\end{split}
\end{equation}
Consequently, \eqref{eq:vsep-general} gives
\begin{equation}\label{eq:general-DiD-variance-gap}
V_{\mathsf{raw}}-V_{\mathsf{sep}}(\eps)
=
\frac{\Gamma_{11}(\eps)+\Gamma_{10}(\eps)}{p_1}
+
\frac{\Gamma_{01}(\eps)+\Gamma_{00}(\eps)}{p_0},
\end{equation}
where $\Gamma_{gt}(\eps)=\Gamma_{F_{gt}}(\eps)$.

\medskip
\noindent\textit{Part \textup{(i)}.}
Index the four distributions by
$(F_1,F_2,F_3,F_4)=(F_{11},F_{10},F_{01},F_{00})$,
with corresponding densities $f_j$ and generic
errors $e_j\sim F_j$.
Set $(w_1,w_2,w_3,w_4)
=(1/p_1,1/p_1,1/p_0,1/p_0)$ and
$D(\eps)=\sum_{j=1}^4w_j\Gamma_{F_j}(\eps)$.
Equation~\eqref{eq:general-DiD-variance-gap} gives
$D(\eps)=V_{\mathsf{raw}}-V_{\mathsf{sep}}(\eps)$.
We prove the assertion by contradiction.
Suppose that, for some $\eps_0\in(0,1/2)$,
$D(\eps)\ge0$ for every $\eps\in(0,\eps_0)$.

Put $\eps=e^{-z}$ and define
\[
q_{j,L}(z)=-F_j^{-1}(e^{-z}),
\qquad
q_{j,R}(z)=F_j^{-1}(1-e^{-z}).
\]
Each $e_j$ is nondegenerate and has mean zero,
so $\Prob(e_j<0)>0$ and $\Prob(e_j>0)>0$.
Choose $z_0$ sufficiently large that
\[
e^{-z_0}
<
\min\left\{
\eps_0,\frac18,
\min_{1\le j\le4}F_j(0),
\min_{1\le j\le4}\{1-F_j(0)\}
\right\}.
\]
Then $q_{j,L}(z),q_{j,R}(z)>0$ for $z\ge z_0$.
The all-fraction cutoff regularity and differentiation
of the inverse distribution functions give
\[
q'_{j,L}(z)
=\frac{e^{-z}}{f_j(-q_{j,L}(z))}>0,
\qquad
q'_{j,R}(z)
=\frac{e^{-z}}{f_j(q_{j,R}(z))}>0.
\]
These derivatives are exactly $r_L$ and $r_R$
for distribution $F_j$ at fraction $e^{-z}$.

Let
$\boldsymbol q(z)
=(q_{1,L}(z),q_{1,R}(z),\ldots,q_{4,L}(z),q_{4,R}(z))$.
For $\boldsymbol x\in\R^8$, define
$\|\boldsymbol x\|_w^2
=\sum_{j=1}^4w_j(x_{j,L}^2+x_{j,R}^2)$.
Write
\[
\begin{aligned}
Q(z)&=\|\boldsymbol q(z)\|_w^2,
&
A(z)&=\int_0^\infty e^{-u}Q(z+u)\,du,\\
B(z)&=A(z)-Q(z),
&
R(z)&=\|\boldsymbol q'(z)\|_w^2.
\end{aligned}
\]

To interpret these integrals, let $W$ have density
$e^{-u}$ on $u\ge0$.
For $x\ge q_{j,R}(z)$ and $y\ge q_{j,L}(z)$,
the inverse-quantile identities give
\[
\begin{split}
\Prob\{q_{j,R}(z+W)>x\}
&=e^z\{1-F_j(x)\}\\
&=\Prob(e_j>x\mid e_j>q_{j,R}(z)),\\
\Prob\{q_{j,L}(z+W)>y\}
&=e^zF_j(-y)\\
&=\Prob(-e_j>y\mid e_j<-q_{j,L}(z)).
\end{split}
\]
Thus these transformed exponential variables have
the corresponding conditional tail distributions.
In particular,
\[
\begin{split}
A(z)=\sum_{j=1}^4w_j\bigl\{
&\E(e_j^2\mid e_j<-q_{j,L}(z))\\
&+\E(e_j^2\mid e_j>q_{j,R}(z))
\bigr\}<\infty.
\end{split}
\]
Define the conditional mean excesses
\[
\begin{split}
m_{j,L}(z)
&=\E(-e_j-q_{j,L}(z)\mid e_j<-q_{j,L}(z)),\\
m_{j,R}(z)
&=\E(e_j-q_{j,R}(z)\mid e_j>q_{j,R}(z)).
\end{split}
\]
Collecting these in a vector $\boldsymbol m(z)$
in the same coordinate order as $\boldsymbol q(z)$,
the preceding distributional identities imply
\[
\boldsymbol m(z)
=\int_0^\infty e^{-u}
\{\boldsymbol q(z+u)-\boldsymbol q(z)\}\,du.
\]
All coordinates are nonnegative and finite.
Jensen's inequality, followed coordinatewise by
$(y-x)^2\le y^2-x^2$ for $y\ge x\ge0$, yields
\begin{equation}\label{ut-eq:m-bound}
\begin{split}
\|\boldsymbol m(z)\|_w^2
&\le
\int_0^\infty e^{-u}
\|\boldsymbol q(z+u)-\boldsymbol q(z)\|_w^2\,du\\
&\le A(z)-Q(z)=B(z).
\end{split}
\end{equation}
Since $Q$ is continuous and
$A(z)=e^z\int_z^\infty e^{-v}Q(v)\,dv<\infty$,
differentiation gives
\begin{equation}\label{ut-eq:B-derivative}
A'(z)=A(z)-Q(z)=B(z),
\qquad
B'(z)=B(z)-Q'(z).
\end{equation}

For the cutoffs
$a_j=-q_{j,L}(z)$ and $b_j=q_{j,R}(z)$,
write $\beta_j(z)=\E\{\phi_{a_j,b_j}(e_j)\}$.
The bias identity \eqref{eq:beta-definition} gives
$\beta_j(z)=e^{-z}\{m_{j,L}(z)-m_{j,R}(z)\}$.
Define
$d_j(z)=q'_{j,R}(z)-m_{j,R}(z)
-q'_{j,L}(z)+m_{j,L}(z)$.
Then \eqref{eq:auxiliary-H} implies
$\eta_{F_j,e^{-z}}=e^{-z}d_j(z)$.
Applying \eqref{eq:general-efficiency-index}, we obtain
\[
\frac{D(e^{-z})}{e^{-z}}
=
A(z)-\|\boldsymbol q(z)+\boldsymbol q'(z)\|_w^2
+e^{-z}\sum_{j=1}^4w_jd_j(z)^2.
\]
Because
$\|\boldsymbol q+\boldsymbol q'\|_w^2
=Q+Q'+R$ and $B'=B-Q'$,
this becomes
\begin{equation}\label{ut-eq:master}
\frac{D(e^{-z})}{e^{-z}}
=
B'(z)-R(z)+e^{-z}\sum_{j=1}^4w_jd_j(z)^2.
\end{equation}
Moreover, using $(x-y)^2\le2x^2+2y^2$ and
\eqref{ut-eq:m-bound},
\begin{equation}\label{ut-eq:d-bound}
\begin{split}
\sum_{j=1}^4w_jd_j(z)^2
&\le2\|\boldsymbol q'(z)-\boldsymbol m(z)\|_w^2\\
&\le4\{R(z)+\|\boldsymbol m(z)\|_w^2\}\\
&\le4\{R(z)+B(z)\}.
\end{split}
\end{equation}
Our assumption $D(e^{-z})\ge0$ therefore implies
\[
B'(z)
\ge R(z)-e^{-z}\sum_{j=1}^4w_jd_j(z)^2
\ge(1-4e^{-z})R(z)-4e^{-z}B(z).
\]
Since $e^{-z}\le1/8$ for $z\ge z_0$,
$B'(z)+4e^{-z}B(z)\ge R(z)/2$.

Define the integrating factor
$J(z)=\exp\{4\int_{z_0}^z e^{-v}\,dv\}$
and $Y(z)=J(z)B(z)$.
Then $1\le J(z)\le J_*=\exp(4e^{-z_0})$ and
\begin{equation}\label{ut-eq:monotoneY}
Y'(z)
=J(z)\{B'(z)+4e^{-z}B(z)\}
\ge\tfrac12J(z)R(z)\ge\tfrac12R(z).
\end{equation}
Since every positive coordinate of $\boldsymbol q$
is strictly increasing,
\[
B(z_0)
=\int_0^\infty e^{-u}\{Q(z_0+u)-Q(z_0)\}\,du>0.
\]
Thus $Y$ is nondecreasing and $Y(z)\ge Y(z_0)>0$.

By the definition of $\mathcal P_{\mathrm{SG}}$
at the beginning of this subsection, for each
$j$ there is $K_j>0$ such that
$\E\exp(se_j)\le\exp(K_j^2s^2/2)$ for all $s\in\R$.
For $x>0$, Markov's inequality gives
\[
\Prob(e_j>x)
\le\inf_{s>0}\exp(-sx+K_j^2s^2/2)
=\exp\{-x^2/(2K_j^2)\}.
\]
Applying the same argument to $-e_j$ gives the
corresponding lower-tail bound.
At the two tail quantiles, each tail probability
equals $e^{-z}$, so
$q_{j,L}(z)^2,q_{j,R}(z)^2\le2K_j^2z$.
Consequently, with $C=4\sum_{j=1}^4w_jK_j^2$,
\[
Q(z)\le Cz,
\qquad
A(z)\le C\int_0^\infty e^{-u}(z+u)\,du=C(z+1).
\]
For $t>z\ge z_0$, monotonicity of $Y$ gives
\[
A(t)-A(z)
=\int_z^t\frac{Y(v)}{J(v)}\,dv
\ge\frac{Y(z)}{J_*}(t-z).
\]
Dividing by $t$ and letting $t\to\infty$ yields
$Y(z)\le CJ_*$.
Integrating \eqref{ut-eq:monotoneY} therefore shows
\[
\int_{z_0}^\infty
\|\boldsymbol q'(v)\|_w^2\,dv
\le2\{CJ_*-Y(z_0)\}<\infty.
\]

For fixed $T\ge z_0$ and $z>T$, the fundamental
theorem of calculus and Cauchy--Schwarz give
\[
\|\boldsymbol q(z)\|_w
\le\|\boldsymbol q(T)\|_w+
\sqrt{z-T}
\left\{\int_T^\infty
\|\boldsymbol q'(v)\|_w^2\,dv\right\}^{1/2}.
\]
Hence
\[
\limsup_{z\to\infty}
\frac{\|\boldsymbol q(z)\|_w}{\sqrt z}
\le
\left\{\int_T^\infty
\|\boldsymbol q'(v)\|_w^2\,dv\right\}^{1/2}
\longrightarrow0
\quad\text{as }T\to\infty.
\]
Thus $Q(z)=o(z)$.
For any $\eta>0$, sufficiently large $z$ satisfy
$Q(z+u)\le\eta(z+u)$ for every $u\ge0$.
It follows that $A(z)\le\eta(z+1)$ and therefore
$A(z)=o(z)$.
On the other hand,
$B(z)=Y(z)/J(z)\ge Y(z_0)/J_*>0$, so
\[
A(z)
=A(z_0)+\int_{z_0}^z B(v)\,dv
\ge A(z_0)+\frac{Y(z_0)}{J_*}(z-z_0).
\]
This contradicts $A(z)=o(z)$.
Therefore, for every $\eps_0\in(0,1/2)$, some
$\eps\in(0,\eps_0)$ satisfies $D(\eps)<0$,
which proves part \textup{(i)}.

\medskip
\noindent\textit{Part \textup{(ii)}.}
Let $P$ denote the common symmetric error
distribution, with distribution function $F$
and density $f$.
Write $q_\eps=F^{-1}(1-\eps)>0$ and
$r_\eps=\eps/f(q_\eps)$.
Symmetry gives $a=-q_\eps$, $b=q_\eps$,
$r_L=r_R=r_\eps$, and $\beta=\eta_{F,\eps}=0$.
Thus, apart from endpoint values,
\begin{equation}\label{eq:cell-if-symmetric}
\psi_{F,\eps}(x)=
\begin{cases}
-(q_\eps+r_\eps),&x<-q_\eps,\\
x,&|x|\le q_\eps,\\
q_\eps+r_\eps,&x>q_\eps.
\end{cases}
\end{equation}
Define
\begin{equation}\label{eq:efficiency-index}
\mathcal G_P(\eps)
=
\E_P(e^2\mid |e|>q_\eps)-(q_\eps+r_\eps)^2.
\end{equation}
Since $\Prob(|e|>q_\eps)=2\eps$,
\eqref{eq:general-efficiency-index} gives
$\Gamma_F(\eps)=2\eps\mathcal G_P(\eps)$.
Equation~\eqref{eq:general-DiD-variance-gap}
therefore reduces to
\begin{equation}\label{eq:variance-difference-identity}
V_{\mathsf{raw}}-V_{\mathsf{sep}}(\eps)
=
4\eps\left(\frac1{p_0}+\frac1{p_1}\right)
\mathcal G_P(\eps).
\end{equation}
Fix an arbitrary $\eps_0\in(0,1/2)$.
It suffices to prove that $\mathcal G_P(\eps)>0$
for some $\eps\in(0,\eps_0)$.
Suppose, to the contrary, that
$\mathcal G_P(\eps)\le0$ for every
$\eps\in(0,\eps_0)$, and set
$q_\star=F^{-1}(1-\eps_0)>0$.
Let $X=|e|$, and define its survival function
$S(q)=\Prob(X>q)=2\{1-F(q)\}$ and density
$g(q)=2f(q)$ for $q>0$.
The strict increase of $F$ and all-fraction
cutoff regularity ensure that $S(q)>0$ and
$g(q)>0$ for every $q>0$.
Define
\[
r(q)=\frac{S(q)}{g(q)},
\qquad
m(q)=\E(X-q\mid X>q).
\]
For each $q>q_\star$, the fraction $\eps=1-F(q)$
belongs to $(0,\eps_0)$, with $q=q_\eps$ and
$r(q)=r_\eps$.
Our assumed inequality therefore gives
$\E(X^2\mid X>q)\le\{q+r(q)\}^2$.
By Cauchy--Schwarz,
\[
q+m(q)
=\E(X\mid X>q)
\le\sqrt{\E(X^2\mid X>q)}
\le q+r(q),
\]
and hence $m(q)\le r(q)$ for every $q>q_\star$.

Set $I(q)=\int_q^\infty S(v)\,dv$.
Since $\E X<\infty$, this integral is finite;
positivity of $S$ also gives $I(q)>0$.
Tail integration yields
\[
I(q)=\E(X-q)_+=S(q)m(q),
\]
so $m(q)=I(q)/S(q)$.
Using $I'(q)=-S(q)$ and $S'(q)=-g(q)$,
we obtain
\[
m'(q)
=\frac{-S(q)^2+I(q)g(q)}{S(q)^2}
=-1+\frac{m(q)}{r(q)}
\le0,\qquad q>q_\star.
\]
Fix $q_0>q_\star$ and write $M=m(q_0)\in(0,\infty)$.
For $q\ge q_0$, $m(q)\le M$, and therefore
\[
\frac{I'(q)}{I(q)}
=-\frac1{m(q)}
\le-\frac1M.
\]
Integration gives
$I(q)\le I(q_0)\exp\{-(q-q_0)/M\}$.
Because $S$ is nonincreasing, for $q\ge q_0+1$,
\[
S(q)
\le\int_{q-1}^{q}S(v)\,dv
\le I(q-1)
\le I(q_0)\exp\{-(q-1-q_0)/M\}.
\]
Thus $S(q)$ has an exponential upper bound.

For any $s\in(0,1/M)$, the identity
$e^{sX}=1+\int_0^X se^{sq}\,dq$ and Tonelli's
theorem give
\[
\E e^{s|e|}
=1+s\int_0^\infty e^{sq}S(q)\,dq<\infty.
\]
The integral over $[0,q_0+1]$ is finite, and the
remaining integral is finite because
$s-1/M<0$.
This contradicts the definition
$\E e^{s|e|}=\infty$ for every $s>0$ when
$P\in\mathcal P_{\mathrm{HT}}$.
Hence $\mathcal G_P(\eps)>0$ for some
$\eps\in(0,\eps_0)$. Since $\eps_0$ was arbitrary,
\eqref{eq:variance-difference-identity} proves
part \textup{(ii)}.
\Halmos\endproof
\subsubsection{Counterexamples to Efficiency Rankings Based
Only on Tail Classes}
\label{app:tail-class-counterexamples}

The following examples show that membership in
$\mathcal P_{\mathrm{SG}}$ or $\mathcal P_{\mathrm{HT}}$
does not guarantee an efficiency ranking at every
winsorization fraction. Fix $\eps_\star=1/40$ and
use the same symmetric error distribution in all four
cells of \eqref{eq:location-shift-asymptotic}.
Both examples then satisfy
(\ref{cond:sep-symmetry}),
(\ref{cond:sep-time-stable}), and
(\ref{cond:sep-group-stable}).

For a symmetric distribution $F$ with density $f$,
write $q_F=F^{-1}(1-\eps_\star)$ and
$h_F=\eps_\star/f(q_F)$.
By the influence function
\eqref{eq:cell-if-symmetric}, the cellwise
asymptotic variance is
\[
v_{\mathsf{win}}(F)
=
\int_{-q_F}^{q_F}x^2f(x)\,dx
+2\eps_\star(q_F+h_F)^2,
\]
whereas $v_{\mathsf{raw}}(F)=\Var_F(e)$.
Theorem~\ref{thm:sep-fixed-clt} gives
$V_{\mathsf{sep}}(\eps_\star)
=2(1/p_0+1/p_1)v_{\mathsf{win}}(F)$,
and the same multiplier relates $V_{\mathsf{raw}}$
to $v_{\mathsf{raw}}(F)$.
Thus it suffices to compare these two cellwise
variances.

\medskip
\noindent\textit{A sub-Gaussian example with
$V_{\mathsf{sep}}(\eps_\star)<V_{\mathsf{raw}}$.}
Let $X=Z+100B$, where $Z\sim N(0,1)$ and,
independently, $B$ takes values $-1,0,1$ with
probabilities $0.01,0.98,0.01$. Its distribution is
$F=0.98N(0,1)+0.01N(-100,1)+0.01N(100,1)$,
which is symmetric with a positive smooth density
and variance $201$. Moreover,
\[
\E\exp(sX)
=\exp(s^2/2)\{0.98+0.02\cosh(100s)\}
\le\exp(10001s^2/2),
\]
using $\cosh(x)\le\exp(x^2/2)$.
Hence $F\in\mathcal P_{\mathrm{SG}}$.

Let $\varphi$ and $\Phi$ denote the standard normal
density and distribution function.
Since
$\Prob(X>3)\le 0.98\{1-\Phi(3)\}+0.02<0.025$,
we have $0<q_F<3$ and
$f(q_F)\ge 0.98\varphi(q_F)\ge 0.98\varphi(3)$.
The influence function
\eqref{eq:cell-if-symmetric} is bounded in absolute
value by $q_F+h_F$, so
\[
v_{\mathsf{win}}(F)
\le
\left(3+\frac{0.025}{0.98\varphi(3)}\right)^2
<78<201=v_{\mathsf{raw}}(F).
\]
The numerical bounds follow from
$\varphi(3)=\exp(-9/2)/\sqrt{2\pi}$ and
$1-\Phi(3)\le\varphi(3)/3$.

\medskip
\noindent\textit{A heavy-tailed example with
$V_{\mathsf{raw}}<V_{\mathsf{sep}}(\eps_\star)$.}
Let $G_3$ denote the Student's $t$ distribution
with three degrees of freedom, whose density is
$g_3(x)=2(\pi\sqrt3)^{-1}(1+x^2/3)^{-2}$.
Take the explicit mixture
$F=(1-\eta)\Phi+\eta G_3$ with $\eta=10^{-6}$.
This distribution is symmetric with a positive
smooth density and variance
$v_{\mathsf{raw}}(F)=1+2\eta=1.000002$.
Its polynomial tails imply
$\E_F\exp(s|X|)=\infty$ for every $s>0$, so
$F\in\mathcal P_{\mathrm{HT}}$.

We bound its winsorized variance directly.
Set $a=1.959$ and $b=1.961$.
Since $\sup_x|F(x)-\Phi(x)|\le\eta$ and
\[
0.974943<\Phi(a)<0.974944,
\qquad
0.975060<\Phi(b)<0.975061,
\]
we have $F(a)<0.975<F(b)$ and hence $a<q_F<b$.
Both $\varphi$ and $g_3$ are decreasing on
$(0,\infty)$, so
$f(q_F)\le(1-\eta)\varphi(a)+\eta g_3(0)$.
Also,
\[
\int_{-q_F}^{q_F}x^2f(x)\,dx
\ge
(1-\eta)\int_{-a}^{a}x^2\varphi(x)\,dx
=
(1-\eta)\{2\Phi(a)-1-2a\varphi(a)\}.
\]
Consequently,
\[
\begin{split}
v_{\mathsf{win}}(F)
\ge{}&
(1-\eta)\{2\Phi(a)-1-2a\varphi(a)\}\\
&+0.05\left(
a+\frac{0.025}
{(1-\eta)\varphi(a)+\eta g_3(0)}
\right)^2\\
>{}&1.005>1.000002=v_{\mathsf{raw}}(F).
\end{split}
\]

Finally, for either example, replace $X$ by
$e=\sigma X/\sqrt{\Var(X)}$.
Positive rescaling preserves the tail class and
multiplies both asymptotic variances by the same
squared scale factor. The resulting error distribution
therefore belongs to $\cP(\sigma^2)$ in
\eqref{eq:P-sigma} and retains the strict comparison.


\subsection{Proof of Theorem~\ref{thm:global-dfr-dominance}}
\label{app:dfr-sufficient}

\proof{Proof.}
Fix one arbitrary cell and $\eps\in(0,1/2)$.
Write $F,f,m$ for its error distribution function, density,
and median. Recall
$a=F^{-1}(\eps)$, $b=F^{-1}(1-\eps)$,
$r_L=\eps/f(a)$, and $r_R=\eps/f(b)$.
By strict monotonicity and the stated cutoff-density conditions,
$a<m<b$, $f(a),f(b)>0$, and
$\Prob(e<a)=\Prob(e>b)=\eps$.

Recall the definitions in \eqref{eq:auxiliary-H}:
\[
H_{F,\eps}(x)=
\begin{cases}
a-r_L,&x<a,\\
x,&a\le x\le b,\\
b+r_R,&x>b,
\end{cases}
\qquad
\eta_{F,\eps}
=\beta+\eps(r_R-r_L),
\]
where $\beta=\E\{\phi_{a,b}(e)\}$.
The representation established in the proof of
Theorem~\ref{thm:tail-class-universal} gives
$\psi_{F,\eps}(e)=H_{F,\eps}(e)-\eta_{F,\eps}$
almost surely and
$\E\{H_{F,\eps}(e)\}=\eta_{F,\eps}$.
Consequently,
\[
\Var\{\psi_{F,\eps}(e)\}
=\Var\{H_{F,\eps}(e)\},
\qquad
\Gamma_F(\eps)
=\Var(e)-\Var\{H_{F,\eps}(e)\},
\]
as in \eqref{eq:general-efficiency-index}.
By \eqref{eq:general-DiD-variance-gap}, it therefore
suffices to prove $\Gamma_F(\eps)>0$ for each cell.

To establish this inequality, define the conditional
mean excesses
$m_L=\E(a-e\mid e<a)$ and
$m_R=\E(e-b\mid e>b)$, and the transformation
\[
M_{F,\eps}(x)=
\begin{cases}
a-m_L,&x<a,\\
x,&a\le x\le b,\\
b+m_R,&x>b.
\end{cases}
\]
Here $a-m_L=\E(e\mid e<a)$ and
$b+m_R=\E(e\mid e>b)$, so $M_{F,\eps}$ replaces
each tail by its conditional mean and leaves the
middle observations unchanged.
We will show that
\[
\Var(e)-\Var\{M_{F,\eps}(e)\}
=
\eps\{\Var(e\mid e<a)+\Var(e\mid e>b)\}>0,
\]
and that the DFR condition implies
\[
m_L\ge r_L,\qquad m_R\ge r_R,
\qquad
\Var\{H_{F,\eps}(e)\}
\le\Var\{M_{F,\eps}(e)\}.
\]
Together, these inequalities give
$\Var\{H_{F,\eps}(e)\}<\Var(e)$, as required.

We first establish the variance reduction from replacing
each tail by its conditional mean. Since
$\Prob(e<a)=\Prob(e>b)=\eps$,
\[
\begin{split}
\E\{M_{F,\eps}(e)\}
&=\eps\,\E(e\mid e<a)
+\E\{e\ind\{a\le e\le b\}\}
+\eps\,\E(e\mid e>b)\\
&=\E e=0.
\end{split}
\]
Thus both $e$ and $M_{F,\eps}(e)$ have mean zero.
Their contributions to the second moment agree on
$[a,b]$, so
\[
\begin{split}
\Var(e)-\Var\{M_{F,\eps}(e)\}
={}&\eps\{\E(e^2\mid e<a)-(a-m_L)^2\}\\
&+\eps\{\E(e^2\mid e>b)-(b+m_R)^2\}\\
={}&\eps\{\Var(e\mid e<a)+\Var(e\mid e>b)\}.
\end{split}
\]
Both conditional variances are finite because
$\E(e^2)<\infty$ and the conditioning events have
positive probability. They are strictly positive
because the positive density near each cutoff makes
the corresponding conditional distribution
nondegenerate. Therefore
$\Var\{M_{F,\eps}(e)\}<\Var(e)$.

We next use the DFR condition to compare
$H_{F,\eps}$ with $M_{F,\eps}$.
Define $R_L$ and $R_R$ by their conditional
distribution functions: for $x\ge0$,
\[
\begin{split}
F_{R_L}(x)
&=\Prob(a-e\le x\mid e<a)
=1-\frac{F(a-x)}{\eps},\\
F_{R_R}(x)
&=\Prob(e-b\le x\mid e>b)
=1-\frac{1-F(b+x)}{\eps},
\end{split}
\]
and set both distribution functions equal to zero
for $x<0$. Then $\E R_L=m_L$ and $\E R_R=m_R$.
Their survival functions are therefore
\[
S_L(x)=\Prob(R_L>x)=\frac{F(a-x)}{\eps},
\qquad
S_R(x)=\Prob(R_R>x)=\frac{1-F(b+x)}{\eps},
\qquad x\ge0.
\]
For $x>0$, differentiation gives the densities
$g_L(x)=f(a-x)/\eps$ and
$g_R(x)=f(b+x)/\eps$.
Their failure rates are consequently
\[
\begin{split}
\lambda_L^a(x)
&=\frac{g_L(x)}{S_L(x)}
=\frac{f(a-x)}{F(a-x)}
=\lambda_L(m-a+x),\\
\lambda_R^b(x)
&=\frac{g_R(x)}{S_R(x)}
=\frac{f(b+x)}{1-F(b+x)}
=\lambda_R(b-m+x),
\end{split}
\]
where $\lambda_L$ and $\lambda_R$ are the outward
failure rates in \eqref{eq:global-asymmetric-dfr}.
Since $a<m<b$, Definition~\ref{def:dfr-tail}
implies that both shifted failure rates are
nonincreasing. Cutoff-density continuity gives
\[
\lim_{x\downarrow0}\lambda_L^a(x)
=\frac{f(a)}{\eps}=\frac1{r_L},
\qquad
\lim_{x\downarrow0}\lambda_R^b(x)
=\frac{f(b)}{\eps}=\frac1{r_R}.
\]
Thus $\lambda_L^a(x)\le1/r_L$ and
$\lambda_R^b(x)\le1/r_R$ for every $x>0$.

Because the survival functions are positive and
absolutely continuous, their logarithmic derivatives
satisfy
$(\log S_L)'=-\lambda_L^a$ and
$(\log S_R)'=-\lambda_R^b$ almost everywhere.
Together with $S_L(0)=S_R(0)=1$, this yields
\[
S_L(x)
=\exp\left\{-\int_0^x\lambda_L^a(v)\,dv\right\}
\ge e^{-x/r_L},
\qquad
S_R(x)
=\exp\left\{-\int_0^x\lambda_R^b(v)\,dv\right\}
\ge e^{-x/r_R}.
\]
Integrating these survival bounds gives
\[
m_L=\E R_L=\int_0^\infty S_L(x)\,dx\ge r_L,
\qquad
m_R=\E R_R=\int_0^\infty S_R(x)\,dx\ge r_R.
\]

We now compare the previously defined transformations
$H_{F,\eps}$ and $M_{F,\eps}$.
The inequalities just established imply
\[
a-m_L\le a-r_L<a<b<b+r_R\le b+m_R.
\]
Thus $H_{F,\eps}$ moves the two constant tail values
of $M_{F,\eps}$ toward $[a,b]$, while leaving the
middle values unchanged. To verify the resulting
contraction, take $x\le y$.
If both points lie in the same tail, both transformed
differences are zero; if both lie in $[a,b]$,
both differences equal $y-x$.
For points in different regions, the differences
under $H_{F,\eps}$ and $M_{F,\eps}$ satisfy
\[
\begin{aligned}
y-a+r_L&\le y-a+m_L,
&&x<a,\quad a\le y\le b,\\
b-x+r_R&\le b-x+m_R,
&&a\le x\le b,\quad y>b,\\
b-a+r_L+r_R&\le b-a+m_L+m_R,
&&x<a,\quad y>b.
\end{aligned}
\]
Since both transformations are nondecreasing,
these comparisons imply
\[
|H_{F,\eps}(x)-H_{F,\eps}(y)|
\le
|M_{F,\eps}(x)-M_{F,\eps}(y)|
\qquad\text{for every }x,y\in\R.
\]

For a square-integrable random variable $Z$
and an independent copy $Z'$, independence gives
$\E(Z-Z')^2=2\Var(Z)$.
Applying this identity with an independent copy
$e'$ of $e$ yields
\[
\begin{split}
\Var\{H_{F,\eps}(e)\}
&=\frac12\E[
\{H_{F,\eps}(e)-H_{F,\eps}(e')\}^2]\\
&\le\frac12\E[
\{M_{F,\eps}(e)-M_{F,\eps}(e')\}^2]\\
&=\Var\{M_{F,\eps}(e)\}.
\end{split}
\]
Combining this inequality with the variance reduction
established above gives
\[
\begin{split}
\Gamma_F(\eps)
&=\Var(e)-\Var\{H_{F,\eps}(e)\}\\
&\ge\Var(e)-\Var\{M_{F,\eps}(e)\}\\
&=\eps\{\Var(R_L)+\Var(R_R)\}>0.
\end{split}
\]
Here $\Var(R_L)=\Var(e\mid e<a)$ and
$\Var(R_R)=\Var(e\mid e>b)$, and their strict
positivity follows from the positive density
near the two cutoffs.

The cell was arbitrary, so
$\Gamma_{gt}(\eps)>0$ for every $g,t\in\{0,1\}$.
Equation~\eqref{eq:general-DiD-variance-gap}
therefore gives
\[
V_{\mathsf{raw}}-V_{\mathsf{sep}}(\eps)
=
\frac{\Gamma_{11}(\eps)+\Gamma_{10}(\eps)}{p_1}
+
\frac{\Gamma_{01}(\eps)+\Gamma_{00}(\eps)}{p_0}
>0,
\]
because $p_0,p_1>0$.
By Proposition~\ref{prop:sep-unbiasedness},
separate winsorization is asymptotically unbiased
under any one of conditions
(\ref{cond:sep-symmetry}),
(\ref{cond:sep-time-stable}), or
(\ref{cond:sep-group-stable}), so the variance
comparison concerns estimators of the same target
$\Delta$. Since $\eps\in(0,1/2)$ was arbitrary,
the result follows.
\Halmos\endproof
\end{APPENDICES}

\end{document}